\documentclass[journal,10pt]{IEEEtaes}

\usepackage{color,array,amsthm}
\usepackage{graphicx}

\jvol{XX}
\jnum{XX}
\jmonth{XXXXX}
\paper{1234567}
\pubyear{2022}
\doiinfo{TAES.2022.Doi Number}

\newtheorem{theorem}{Theorem}
\newtheorem{lemma}{Lemma}
\usepackage{amsfonts,amsmath,amssymb,amsthm}
\usepackage[ruled,lined,linesnumbered]{algorithm2e}
\usepackage{amsbsy,caption}
\usepackage[para]{threeparttable}
\usepackage{graphicx,multirow,bm}
\usepackage{fancybox}
\usepackage{scalefnt}
\usepackage[noadjust]{cite}
\usepackage{tikz}
\usepackage{color}
\makeatletter
\newtheoremstyle{mythm}{3pt}{3pt}{}{16pt}{\bfseries}{:}{.5em}{}
\theoremstyle{mythm}
\newtheorem{example}{Example}
\newtheorem{definition}{Definition}
\newtheorem{remark}{Remark}

\renewcommand\thetable{\Roman{table}}

\renewcommand{\leq}{\leqslant}

\renewcommand{\geq}{\geqslant}

\usepackage{subfigure}
\usepackage{graphicx}
\usepackage{multirow}
\usepackage{cite}
\usepackage{epsfig}
\usepackage{epstopdf}

\begin{document}

	\title{Enhanced Low-Redundancy Restricted Array for Direction of Arrival Estimation}

	\author{Shidong Zhang}
	\affil{Southwest Jiaotong University, Chengdu, 611756, China}
	
	\author{Zhengchun Zhou}
	\member{Member, IEEE}
	\affil{Southwest Jiaotong University, Chengdu, 611756, China}
	
	\author{Guolong Cui}
	\member{Senior Member, IEEE}
	\affil{University of Electronic Science and Technology of China, Chengdu, 610054, China}
	
	\author{Xiaohu Tang}
	\member{Senior Member, IEEE}
	\affil{Southwest Jiaotong University, Chengdu, 611756, China}
	
	\author{Pingzhi Fan}
	\member{Fellow, IEEE}
	\affil{Southwest Jiaotong University, Chengdu, 611756, China}
	
	
	
	
	\authoraddress{Shidong Zhang is with
		the School of Mathematics, Southwest Jiaotong University, Chengdu, 611756, China (e-mail: 838412679@qq.com). Zhengchun Zhou, Xiaohu Tang, and Pingzhi Fan are the School of Information
		Science and Technology, Southwest Jiaotong University, Chengdu, 611756, China (e-mail: zzc@swjtu.edu.cn, xhutang@swjtu.edu.cn, pzfan@swjtu.edu.cn). Guolong Cui is the School of Information and Communication Engineering, University of Electronic Science and Technology of China, Chengdu, 610054, China (email: cuiguolong@uestc.edu.cn). \textit{THIS PAPER IS SUBMITTED TO IEEE TRANSACTIONS ON AEROSPACE AND ELECTRONIC SYSTEMS FOR POSSIBLE PUBLICATION.}}

	\markboth{AUTHOR ET AL.}{SHORT ARTICLE TITLE}
	\maketitle

	\begin{abstract}
		Sensor arrays play a significant role in direction of arrival (DOA) estimation. Specifically, arrays with low redundancy and reduced mutual coupling are desirable. In this paper, we investigate a sensor array configuration that has a restricted sensor spacing and propose a closed-form expression. We also propose several classes of low redundancy (LR) arrays. Interestingly, compared with super nested arrays (SNA) and maximum inter-element spacing constraint (MISC) arrays, one of the proposed arrays has a significant reduction in both redundancy ratio and mutual coupling. Numerical simulations are also conducted to verify the superiority of the proposed array over the known sparse arrays in terms of weight functions, mutual coupling matrices as well and DOA estimation performance.

	\end{abstract}
	
	
	\section{Introduction}

	{Array signal processing plays a important role in many applications such as radar, sonar, navigation, wireless communications, electronic surveillance and radio astronomy \cite{1984Haykin,2001Sklnik,2016Balanis,1996Krim,2006Peng}.} Key benefits of using sensor arrays include spatial selectivity and the capability to mitigate interference and improve signal quality. {The traditional \textit{uniform linear array} (ULAs) is the most commonly used sensor array, in which the distance between adjacent sensors is less than half of wavelength to avoid spatial aliasing.} However, an $N$-sensor ULA can resolve only up to $N-1$ sources by the traditional subspace-based methods \cite{1989Roy,1986Schmidt}. Besides, electromagnetic characteristics cause {\it mutual coupling} between sensors, which becomes larger as the distance between sensors decreases. {Hence, there is a serious mutual coupling effect in ULAs, which results in an adverse effect on the performance of DOA estimation.}
	
	%
	
	

	To overcome these problems caused by traditional ULAs, {\it nonuniform linear arrays} (NLAs) (also referred as {\it sparse arrays}) were introduced \cite{2001Sklnik}.
	An $N$-sensor NLA is often represented by its sensor locations $n_i\lambda/2$, where $n_i$ belongs to an integer set
	\[{\mathbb S}=\{s_1,s_2,\ldots,s_N\},\]
	and $\lambda$ is the wavelength.
	%
	The {\it difference co-array} (DCA) of the NLA is defined as the array which has sensors located at  positions given by the set ${\cal D}=\{s_i-s_j:i,j=1,2,\ldots,N\}$, and the cardinality of {\it maximal} central ULA segment in DCA is called {\it uniform degrees of freedom} (uDOF), which can be used to detect uncorrelated sources in signal processing techniques, like  co-array MUSIC. Generally, if the uDOF is equal to $U$, then the number of uncorrelated sources that can be detected by using co-array MUSIC is up to $(U-1)/2$ \cite{2010Pal}, \cite{2011Pal}.
	%
	Note that there are totally $N^2$ elements in ${\cal D}$, although some locations maybe repeated. {There is a possibility that we can get ${\cal O}(N^2)$ uDOFs using only $N$ physical sensors, and the result is a significant increase in the number of detectable signal sources} \cite{1990Hoctor}. In fact, many works have been done for this purpose \cite{2010Pal}, \cite{2011Vaid}. {Moreover, the reduction of smaller sensor spacing in sparse array will also reduce the mutual coupling effect.} \cite{2017BouDaher}. For simplicity, we will use ${\mathbb S}$ and ${\cal D}$ to denote a sparse array and its corresponding DCA respectively, in the rest of this paper.
	
	A sparse array is called {\it restricted} (or {\it hole-free}) if its DCA is a ULA, i.e., ${\cal D}=[-L,L]$, where $L$ denotes the maximal inter-sensor spacing (called {\it the aperture} of the array). For example, $\{0,1,4,6\}$ is a 4-sensor restricted array since its DCA ${\cal D}=[-6,6]$. In contrast, a sparse array is called {\it general} if its DCA is not a ULA.  The array $\{0,1,4\}$ is such an example since the spacing $2$ is missing in its DCA. In this paper, we consider only the \textit{restricted problem}. Unless stated otherwise, the term ``array" will always refer to a restricted array. About the research on general arrays, interested readers can refer to \cite{2011Pal,1968Moffet,2022Peng,2015Qin,2019Raza,2022ShiJ,2021Shiwanlu,2022Shi,2013Zhang,2019Zhengw,2020Zheng,1977Bloom,1986Ver} for details.

	Judging an array is good or bad, one criterion is the {\it redundancy ratio} $R$ of ${\mathbb S}$ which is quantitatively defined as the number of pairs of sensors divided by $L$: $R=N(N-1)/(2L)$ \cite{1968Moffet}. Obviously $R\geq1$. Bracewell \cite{1966Bracewell} proved that arrays with $R=1$ (``zero-redundancy"  arrays) only exist for $N\leq 4$. In 1956, Leech \cite{1956Leech} demonstrated that for {\it minimum redundancy arrays} (MRA) (achieving the lowest $R$): $1.217\leq R_{opt}\leq 1.674$ for $N\rightarrow\infty$, and provided some optimal solutions for $N\leq 11$, which was expanded to $N\leq 26$ by an exhaustive computer search in \cite{2020Schwartau}.
	For larger value of $N$, the optimal design of such arrays is not easy, and in most cases, they are restricted to complicated algorithms for sensor placement \cite{1968Moffet}, \cite{1980Ishiguro}. Due to the difficulty of MRAs, several early attempts have been made to construct large {\it low-redundancy linear arrays} (LRAs) (approaching Leech's lower bound).

	Take into consideration the practical application, another criterion of an array is the weight function $\omega(d)$ which is defined as the number of sensor pairs that lead to coarray index $d\in{\cal D}$: $\omega(d)=|\{(n_i,n_j)\in{\mathbb S}^2:n_i-n_j=d\}|$. The most common operation is to decrease the first three weight functions $\omega(1),\omega(2)$ and $\omega(3)$, which cause the most severe mutual coupling between sensors.
	

	According to the two criteria, how to design a sensor array with low-redundancy and reduced mutual coupling becomes a fundamental problem in array signal processing. In addition, for practical applications, the sensor array should be described for any $N$ using {\it simple rules} or {\it closed-form expressions}. Another common representation of an sensor array is described in terms of its inter-sensor spacings, i.e., ${\mathbb D}=\{d_1,d_2,\ldots,d_{N-1}\}$ with $d_i=s_{i+1}-s_i$ for any $1\leq i\leq N-1$. For simplicity, ${\mathbb S}$ and ${\mathbb D}$ will be referred to be the location array and the spacing array respectively in this paper.

	Design of LRAs is equivalent to that of finding ``restricted difference bases" in number theory, and a lot of rich conclusions are obtained in literature \cite{1956Leech,1957Hase,1963Wichmann,1967Wild,1971Miller}. To have a better comparison, we give a summary of the best array structures for known LRAs in Table \ref{table:redundancy thickapprox2}.

	Seen from  Table \ref{table:redundancy thickapprox2}, all spacing arrays of known good LRAs listed coincide with a common pattern:
	\begin{align}\label{general structure}
		{\mathbb D}=\{a_1,a_2,\ldots,a_{s_1},c^\ell,b_1,b_2,\ldots,b_{s_2}\}
	\end{align}
	with the restriction
	\begin{align}\label{Dong-condition-1}
		s_1+s_2=c-1.
	\end{align}
	Here $c^\ell$ denotes the largest spacing $c$ (called {\it base} of the pattern) repeats $\ell$ times \cite{2002Camps}. This pattern was first found by Ishiguro \cite{1980Ishiguro}, and summarized by Dong {\it et al} \cite{2010Dong}. It is known that the $(4r+3)$-Type-93 array in \cite{1993Linebarger} achieves the lowest redundancy  ratio $R\leq 1.5$ (or the highest uDOFs) for fixed sensor number $N$ up to now, but it contains dense ULAs which result in severe mutual coupling; Although the $(4r)$-Type array-93 ($r$ is a natural number) in  the same paper has a smaller uDOFs than the $(4r+3)$-Type-93 array, its mutual coupling can be significantly reduced since $\omega(1)=2$, which may be more practical in application. Inspired by the  nested array (NA) proposed in \cite{2010Pal}, many NA-like arrays with $R\approx2$ were successively proposed, such as
	the improved nested array (INA) \cite{2016Yang}, the super nested array (SNA) \cite{2016Liu-1}, \cite{2016Liu-2}, the augmented nested array (ANA) \cite{2017Liu}, and the generalised extended nested array with multiple subarrays (GENAMS) \cite{2023Wandale}. Especially the SNA can achieve $\omega(1)=1$. Recently an array configuration with $R\approx 2$ based on the maximum inter-element spacing constraint (MISC) criterion was introduced in
	\cite{2019Zheng}, which can provide a higher uDOFs than all the NA-like arrays and still maintain $\omega(1)=1$.
	
	Motivated by the works of \cite{2010Dong,1957Hase,2019Zheng}, we search for some array designs which can result to arrays with low redundancy ratio as well as low mutual coupling.
	We still adopt the common pattern given in (\ref{general structure}), but with a different restriction as
	\begin{align}\label{Our-condition}
		s_1+s_2=c.
	\end{align}
	By careful study on this restriction, 3 types of LRAs with $R\approx 1.5$ are obtained for any $N\geq18$, dependent on how the base of the array reduces mod 4. The detailed representations of our spacing arrays are listed in Table \ref{new repre}.
	%
	%
	Compared with known arrays with the same type, the proposed $(4r+1)$-Type arrays and the $(4r)$-Type arrays all achieve the lowest mutual coupling by decreasing the corresponding $\omega(1), \omega(2)$ and $\omega(3)$, and their uDOFs are at most 4 less for any $N\geq18$.
	Especially, our  $(4r)$-Type array can decrease $\omega(1)$ to the lowest number 1, thus it is the first class of sensor arrays achieving  $R< 1.5$ and $\omega(1)=1$ for any $N\geq18$.
	
	The contributions are summarized as below:
	\begin{itemize}
		\item We generalize the common patterns of existing arrays and provide different restriction $s_1+s_2=c$.
		\item Under the provided restrictions, We obtain 2 classes of $(4r+3)$-type arrays, 5 classes of $(4r+1)$-type arrays, and 3 class of $(4r)$-type arrays for any $N\geq18$.
		\item We analyze the properties of the proposed arrays, including uDOFs and weight functions, and make comparisons existing arrays.
	\end{itemize}
	Finally, simulations validated the effectiveness and superiority of the proposed arrays in DOA estimation.

	\begin{table*}[htbp]
		\centering
		\tiny{
			\caption{A SUMMARY OF LOW-REDUNDANCY LINEAR ARRAYS\label{table:redundancy thickapprox2}}
			\newcommand{\tabincell}[2]{\begin{tabular}{@{}#1@{}}#2\end{tabular}}
				\begin{threeparttable}          
					\begin{tabular}{|c|c|c|c|c|c|}
						\hline
						Refer.  & Spacing Array Structure & uDOFs  & Weight Function  \cr\hline
						\tabincell{c}{Wichmann\\ (1962, \cite{1963Wichmann})} & \tabincell{c}{$\{1^{r},(2r+2)^{r+1},(4r+3)^{2r+1},(2r+1)^{r},(r+1),1^{r}\}$\\$N=6r+4,$$r\geq1$} & $\frac{2N^2+2N-1}{3}$  & $\begin{cases}w(1)=2r\\w(2)=2r-2\\w(3)=2r-4\end{cases}$ \cr\hline
						\tabincell{c}{Bracewell\\(1966, \cite{1966Bracewell})}  & \tabincell{c}{$\begin{cases}
								\{1^{m-1},(m+1),(m)^{m-1}\}, & \mbox{if }N\overset{2}{\equiv}0 \\
								\{1^m,(m+2),(m+1)^{m-1}\}, & \mbox{if }N\overset{2}{\equiv}1
							\end{cases}$} & $\begin{cases}\frac{N^2}{2}+N+1, & \mbox{if }N\overset{2}{\equiv}0\\\frac{N^2}{N}+N+\frac{3}{2}, & \mbox{if }N\overset{2}{\equiv}1\end{cases}$ & $\begin{cases}w(1) & =
							\left\lfloor\frac{N-1}{2}\right\rfloor\\
							w(2) & =
							\left\lfloor\frac{N-1}{2}\right\rfloor-1\\
							w(3) & =
							\left\lfloor\frac{N-1}{2}\right\rfloor-2\\\end{cases}$ \cr\hline
						\tabincell{c}{$(4r+3)$-Type-93\\(1993, \cite{1993Linebarger})} & \tabincell{c}{$\{1^{r},(2r+2)^{r+1},(4r+3)^{2r+k-3},(2r+1)^{r},(r+1),1^{r}\}$\\$N=6r+k,r\geq1,1\leq k\leq6$} & $\frac{2N^2+2N-2k^2+16k-33}{3}$  & $\begin{cases}w(1)=2r\\w(2)=2r-2\\w(3)=2r-4\end{cases}$ \cr\hline
						\tabincell{c}{$(4r+1)$-Type-93\\(1993, \cite{1993Linebarger})} & \tabincell{c}{$\{1^{r},(2r+1)^{r},(4r+1)^{2r+k-1},(2r)^{r-1},r,1^{r}\}$\\$N=6r+k,r\geq2,-1\leq k\leq4$}  & $\frac{2N^2-2k^2+6k-3}{3}$  & $\begin{cases}w(1)=2r\\w(2)=2r-2\\w(3)=2r-4\end{cases}$ \cr\hline
						\tabincell{c}{$(4r)$-Type-93\\ (1993, \cite{1993Linebarger})} & \tabincell{c}{$\{1,2^{r},(2r+1)^{r-1},(4r)^{2r+k},(2r-1)^{r-1},1,2^{r-1}\}$\\$N=6r+k,r\geq2,-3\leq k\leq2$}  & $\frac{2N^2-2k^2+3}{3}$  & $\begin{cases}w(1)=2\\w(2)=2r-1\\w(3)=2\end{cases}$ \cr\hline
						\tabincell{c}{$(4r+3)$-Type-10\\ (2010, \cite{2010Dong})} & \tabincell{c}{$\{1^{r},(2r+2)^{r},(4r+3)^{2r+k-3},(2r+1)^{r+1},(r+1),1^{r}\}$\\$N=6r+k,r\geq1,1\leq k\leq6$} & $\frac{2N^2+2N-2k^2+16k-39}{3}$  & $\begin{cases}w(1)=2r\\w(2)=2r-2\\w(3)=2r-4\end{cases}$ \cr\hline
						\tabincell{c}{$(4r+1)$-Type1-10\\ (2010, \cite{2010Dong})} & \tabincell{c}{$\{1^{r},(2r),1,(2r+1)^{r-1},(4r+1)^{2r+k-1},(2r)^{r-1},(r+1),1^{r-1}\}$\\$N=6r+k,r\geq2,-1\leq k\leq4$}  & $\frac{2N^2-2k^2+6k-3}{3}$  & $\begin{cases}w(1)=2r-1\\w(2)=2r-3\\w(3)=2r-5\end{cases}$ \cr\hline
						\tabincell{c}{$(4r+1)$-Type2-10\\ (2010, \cite{2010Dong})} & \tabincell{c}{$\{1^{r},(2r+1)^{r-1},(4r+1)^{2r+k-1},(2r)^{r},r,1^{r}\}$\\$N=6r+k,r\geq2,-1\leq k\leq4$}  & $\frac{2N^2-2k^2+6k-9}{3}$  & $\begin{cases}w(1)=2r\\w(2)=2r-2\\w(3)=2r-4\end{cases}$ \cr\hline
						\tabincell{c}{$(4r)$-Type-10\\ (2010, \cite{2010Dong})} & \tabincell{c}{$\{1^{2},2^{r-2},1,(2r-1)^{r-1},(4r)^{2r+k},(2r+1)^{r-2},2,(2r+1),2^{r-1}\}$\\$N=6r+k,r\geq3,-3\leq k\leq2$} & $\frac{2N^2-2k^2-3}{3}$  & $\begin{cases}w(1)=3\\w(2)=2r-1\\w(3)=2\end{cases}$ \cr\hline
						\tabincell{c}{NA\\(2010, \cite{2010Pal})}  & \tabincell{c}{$\begin{cases}
								\{1^m,(m+1)^{m-1}\}, & \mbox{if }N\overset{2}{\equiv}0 \\
								\{1^m,(m+1)^{m}\}, & \mbox{if }N\overset{2}{\equiv}1
							\end{cases}$} & $\begin{cases}\frac{N^{2}}{2}+N-1, & \mbox{if }N\overset{2}{\equiv}0\\\frac{N^{2}}{2}+N-\frac{1}{2}, & \mbox{if }N\overset{2}{\equiv}1\end{cases}$ & $\begin{cases}w(1) & =
							\left\lfloor\frac{N}{2}\right\rfloor\\
							w(2) & =
							\left\lfloor\frac{N-1}{2}\right\rfloor-1\\
							w(3) & =
							\left\lfloor\frac{N-1}{2}\right\rfloor-2\\\end{cases}$ \cr\hline
						\tabincell{c}{INA\\(2016, \cite{2016Yang})}  & \tabincell{c}{$\begin{cases}
								\{1^{m-1},(m+1)^{m-1},m\}, & \mbox{if }N\overset{2}{\equiv}0 \\
								\{1^{m-1},(m+1)^{m},m\}, & \mbox{if }N\overset{2}{\equiv}1
							\end{cases}$} & $\begin{cases}\frac{N^{2}}{2}+2N-3,& \mbox{if $N\overset{2}{\equiv}0$}\\\frac{N^{2}}{2}+2N-\frac{7}{2},& \mbox{if $N\overset{2}{\equiv}1$}\end{cases}$  & $\begin{cases}w(1) & =
							\left\lfloor\frac{N}{2}\right\rfloor-1\\
							w(2) & =
							\left\lfloor\frac{N}{2}\right\rfloor-2\\
							w(3) & =
							\left\lfloor\frac{N}{2}\right\rfloor-3\\\end{cases}$ \cr\hline
						\tabincell{c}{SNA\\ (2016, \cite{2016Liu-2})}  & \tabincell{c}{$\begin{cases}
								\{2^{\frac{m-1}{2}},3,2^{\frac{m-1}{2}},(m+1)^{m-3},m,1\}, & \mbox{if }N\overset{2}{\equiv}0,\frac{N}{2}\mbox{is odd} \\
								\{2^{\frac{m-1}{2}},3,2^{\frac{m-1}{2}},(m+1)^{m-2},m,1\}, & \mbox{if }N\overset{2}{\equiv}1,\frac{N-1}{2}\mbox{is odd} \\
							\end{cases}$} & $\begin{cases}\frac{N^{2}}{2}+N-1, & \mbox{if }N\overset{2}{\equiv}0\\\frac{N^{2}}{2}+N-\frac{1}{2}, & \mbox{if }N\overset{2}{\equiv}1\end{cases}$ & $\begin{cases}w(1) & =1\\
							w(2) & =\left\lfloor\frac{N}{2}\right\rfloor-1\\
							w(3) & =1\end{cases}$ \cr\hline
						\tabincell{c}{ANAI-2\\ (2017, \cite{2017Liu})}  & \tabincell{c}{$\begin{cases}
								\{1,2^{\frac{m}{2}-1},1,(m+1)^{m-1},2^{\frac{m}{2}-1}\}, & \mbox{if $N\overset{4}{\equiv}0,m=\frac{N}{2}$} \\
								\{1,2^{\frac{m}{2}-1},1,(m+1)^{m},2^{\frac{m}{2}-1}\}, & \mbox{if $N\overset{4}{\equiv}1,m=\frac{N-1}{2}$}\\
								\{1,2^{\frac{m}{2}},(m+1)^{m},1,2^{\frac{m}{2}-1}\}, & \mbox{if $N\overset{4}{\equiv}2,m=\frac{N-2}{2}$}\\
								\{1,2^{\frac{m}{2}},(m+1)^{m+1},1,2^{\frac{m}{2}-1}\}, & \mbox{if $N\overset{4}{\equiv}3,m=\frac{N-3}{2}$}\\
							\end{cases}$} & $\begin{cases}\frac{N^{2}}{2}+2N-5,& \mbox{if $N\overset{2}{\equiv}0$}\\\frac{N^{2}}{2}+2N-\frac{11}{2},& \mbox{if $N\overset{2}{\equiv}1$}\end{cases}$  & $\begin{cases}w(1) & =2\\
							w(2) & =
							\begin{cases}
								\left\lfloor\frac{N}{2}\right\rfloor-2 & \mbox{if $N\overset{4}{\equiv}0,1$}\\
								\left\lfloor\frac{N-2}{2}\right\rfloor-1 & \mbox{if $N\overset{4}{\equiv}2,3$}\\
							\end{cases}\\
							w(3) & =2\end{cases}$ \cr\hline
						
						\tabincell{c}{MISC\\ (2019, \cite{2019Zheng})} & \tabincell{c}{$\{1,P-3,P^{N-P},2^\frac{P-4}{2},3,2^{\frac{P-4}{2}}\},$ $P=2\left\lfloor\frac{N}{4}\right\rfloor+2$} & $\begin{cases}\frac{N^{2}}{2}+3N-8.5, & \mbox{if $N\overset{4}{\equiv}1$}\\\frac{N^{2}}{2}+3N-9,& \mbox{if $N\overset{2}{\equiv}0$}\\\frac{N^{2}}{2}+3N-10.5,& \mbox{if $N\overset{4}{\equiv}3$}\end{cases}$  & $\begin{cases}w(1)=1\\w(2)=2\left\lfloor\frac{N}{4}\right\rfloor-2\\w(3)=\begin{cases}1,&\mbox{if}N\neq9\\2,& \mbox{if}N=9\end{cases}\end{cases}$ \cr\hline
					\end{tabular}
					\begin{tablenotes}    
						\footnotesize               
						\item[1] The symbol $N\overset{N_{1}}{\equiv}N_{2}$ represents $N\equiv N_{2}\pmod {N_{1}}$.
					\end{tablenotes}
			\end{threeparttable}}          
		\end{table*}
		
		\begin{table*}[htbp]
			\centering
			\caption{A SUMMARY OF THE PROPOSED LOW-REDUNDANCY LINEAR ARRAYS}\label{new repre}
			\newcommand{\tabincell}[2]{\begin{tabular}{@{}#1@{}}#2\end{tabular}}
			\scalebox{0.75}{
				\begin{tabular}{|c|c|c|c|c|c|}
					\hline
					Refer.  & Spacing Array Structure & uDOFs  & Weight Function  \cr\hline
					\tabincell{c}{$(4r+3)$-Type 1} & \tabincell{c}{$\{r+1,1^{r},(2r+2)^{r+1},(4r+3)^{2r+k-4},(2r+1)^{r},r+1,1^{r}\}$\\$N=6r+k,r\geq1,2\leq k\leq7$} & $\frac{2N^2-N-2k^2+19k-45}{3}$  & $\begin{cases}w(1)=2r\\w(2)=2r-2\\w(3)=2r-4\end{cases}$ \cr\hline
					\tabincell{c}{$(4r+3)$-Type 2} & \tabincell{c}{$\{1^{r},(2r+2)^{r+1},(4r+3)^{2r+k-4},(2r+1)^{r},r+1,1^{r},r+1\}$\\$N=6r+k,r\geq1,2\leq k\leq7$} & $\frac{2N^2-N-2k^2+19k-45}{3}$  & $\begin{cases}w(1)=2r\\w(2)=2r-2\\w(3)=2r-4\end{cases}$ \cr\hline
					\tabincell{c}{$(4r+1)$-Type 1} & \tabincell{c}{$\{1^{r-1},(2r+1)^{r+1},(4r+1)^{2r+k-2},(2r)^{r},r-1,2,1^{r-1}\}$\\$N=6r+k,r\geq2,-1\leq k\leq4$} & $\frac{2N^2-2k^2+6k-9}{3}$  & $\begin{cases}w(1)=2r-2\\w(2)=2r-4\\w(3)=2r-6\end{cases}$ \cr\hline
					\tabincell{c}{$(4r+1)$-Type 2} & \tabincell{c}{$\{1^{r-1},r-1,r+2,(2r+1)^{r},(4r+1)^{2r+k-2},(2r)^{r},r+1,1^{r-1}\}$\\$N=6r+k,r\geq3,-1\leq k\leq4$} & $\frac{2N^2-2k^2+6k-9}{3}$  & $\begin{cases}w(1)=2r-2\\w(2)=2r-4\\w(3)=2r-6\end{cases}$ \cr\hline
					\tabincell{c}{$(4r+1)$-Type 3} & \tabincell{c}{$\{1^{r-1},r,r+1,(2r+1)^{r},(4r+1)^{2r+k-2},(2r)^{r},r+1,1^{r-1}\}$\\$N=6r+k,r\geq2,-1\leq k\leq4$} & $\frac{2N^2-2k^2+6k-9}{3}$  & $\begin{cases}w(1)=2r-2\\w(2)=2r-4\\w(3)=2r-6\end{cases}$ \cr\hline
					\tabincell{c}{$(4r+1)$-Type 4} & \tabincell{c}{$\{1^{r-1},3,2r-2,(2r+1)^{r},(4r+1)^{2r+k-2},(2r)^{r},r+1,1^{r-1}\}$\\$N=6r+k,r\geq3,-1\leq k\leq4$} & $\frac{2N^2-2k^2+6k-9}{3}$  & $\begin{cases}w(1)=2r-2\\w(2)=2r-4\\w(3)=2r-5\end{cases}$ \cr\hline
					\tabincell{c}{$(4r+1)$-Type 5} & \tabincell{c}{$\{1^{r-1},(2r+1)^{r+1},(4r+1)^{2r+k-2},(2r)^{r},1,r,1^{r-1}\}$\\$N=6r+k,r\geq2,-1\leq k\leq4$} & $\frac{2N^2-2k^2+6k-9}{3}$  & $\begin{cases}w(1)=2r-1\\w(2)=2r-4\\w(3)=2r-6\end{cases}$ \cr\hline
					\tabincell{c}{$(4r)$-Type 1} & \tabincell{c}{$\{1,2^{r-2},3,(2r-1)^{r},(4r)^{2r+k-1},(2r+1)^{r-1},2,2r-1,2^{r-1}\}$\\$N=6r+k,r\geq3,-2\leq k\leq3$} & $\frac{2N^2-2k^2-9}{3}$  & $\begin{cases}w(1)=1\\w(2)=2r-2\\w(3)=2\end{cases}$ \cr\hline
					\tabincell{c}{$(4r)$-Type 2} & \tabincell{c}{$\{1,2^{r-1},1,(2r-1)^r,(4r)^{2r+k-1},(2r+1)^r,2^{r-1}\}$\\$N=6r+k,r\geq3,-2\leq k\leq3$} & \tabincell{c}{$\frac{2N^2-2k^2-9}{3}$}  & $\begin{cases}w(1)=2\\w(2)=2r-2\\w(3)=2\end{cases}$ \cr\hline
					\tabincell{c}{$(4r)$-Type 3} & \tabincell{c}{$\{2^{r-1},1,(2r-1)^r,(4r)^{2r+k-1},(2r+1)^r,2^{r-1},1\}$\\$N=6r+k,r\geq3,-2\leq k\leq3$} & $\frac{2N^2-2k^2-9}{3}$  & $\begin{cases}w(1)=2\\w(2)=2r-2\\w(3)=2\end{cases}$ \cr\hline
			\end{tabular}}
		\end{table*}

				{The remaining paper is organized as follows. Some necessary preliminaries are introduced in Section \ref{preliminary}. In Section \ref{New array}, we present the array configurations under the different restriction, include their design rules and array structures, and analyze the numbers of uDOFs, and the weight functions in different arrays. Numerical results are shown in Section \ref{numerical examples}. Finally, conclusions are given in Section \ref{conclusion}.}

				\section{Preliminaries}\label{preliminary}
				\subsection{Difference Co-array Signal Model}
				
				Consider an $N$-sensor NLA whose sensor positions are given by ${\mathbb S}=\{s_1,s_2,\ldots,s_N\}$, with difference co-array ${\cal D}$. Assume that $K$ far-field, uncorrelated narrowband signals impinge on the array from distinct directions $\{\theta_1,\theta_2,\ldots,\theta_K\}$ with powers $\{\sigma_1^2,\sigma_2^2,\ldots,\sigma_K^2\}$. {Receiving signal of the array at time $t$ can be expressed as}
				\begin{align}\label{receive signal}
					{\bf x}(t)=\sum\limits_{k=1}^K{\bf a}(\bar{\theta}_k)s_k(t)+{\bf n}(t)={\bf A}{\bf s}(t)+{\bf n}(t),
				\end{align}
				where ${\bf s}(t)=[s_1(t),s_2(t),\ldots,s_K(t)]^T$ is the signal waveform vector, ${\bf A}=[{\bf a}(\bar{\theta}_1),{\bf a}(\bar{\theta}_2),\ldots,{\bf a}(\bar{\theta}_K)]$ is the $N\times K$ array manifold matrix, and ${\bf a}(\bar{\theta}_k)=[e^{j2\pi s_1\bar{\theta}_k},e^{j2\pi s_2\bar{\theta}_k},\ldots,e^{j2\pi s_N\bar{\theta}_k}]^T$ is the steering vector of the array corresponding to the $k$-th signal with $\bar{\theta}_k=\sin{\theta_k}/2$ denoting the normalized DOA satisfying $-1/2\leq \bar{\theta}_k\leq 1/2$. The noise ${\bf n}(t)$ is assumed to be temporally and spatially white, and uncorrelated from the sources.
				The $\bar{\theta}_k$ is considered to be fixed but unknown, and estimated by the signal model (i.e., DOA estimation).
				
				{The covariance matrix of ${\bf x}(t)$ can be approximated as}
				\begin{align}\label{covariance matrix}
					{\bf R}_{\bf xx}&=E[{\bf x}(t){\bf x}^H(t)]\notag\\
					&={\bf A}\mathrm{diag}([\sigma_1^2,\ldots,\sigma_K^2]){\bf A}^H+\sigma_n^2{\bf I}_N\notag\\
					&=\sum\limits_{k=1}^K\sigma_k^2{\bf a}(\bar{\theta}_k){\bf a}^H(\bar{\theta}_k)+\sigma_n^2{\bf I}_N,
				\end{align}
				{where $\sigma_n^2$ is the noise variance.} Since the entries in ${\bf a}(\bar{\theta}_k){\bf a}^H(\bar{\theta}_k)$ are of the form $e^{j2\pi\bar{\theta}_kd}$ for $d\in {\cal D}$, it enables us to reshape (\ref{covariance matrix}) into an autocorrelation vector ${\bf z}$ as in \cite{2010Pal}, \cite{2015Liu}
				\begin{align}\label{vector}
					{\bf z}&=\sum\limits_{k=1}^K\sigma_k^2{\bf b}(\bar{\theta}_k)+\sigma^2_n{\bf e}_0\notag\\
					&={\bf B}{\bf p}+\sigma^2_n{\bf e}_0,
				\end{align}
				where ${\bf b}(\bar{\theta}_k)=[e^{j2\pi d\bar{\theta}_k}]^T_{d\in{\cal D}}$, ${\bf B}=[{\bf b}(\bar{\theta}_1),\ldots,{\bf b}(\bar{\theta}_K)]$, ${\bf p}=[\sigma_1^2,\ldots,\sigma_K^2]$ and $\langle {\bf e}_0\rangle_d=\delta_{d,0}$ for $d\in {\cal D}$. Here $\delta_{p,q}$ is the Kronecker delta. {Comparing (\ref{receive signal}) with (\ref{vector}), the vector ${\bf z}$ can be observed as the received data from a coherent source signal vector ${\bf p}$ with a single snapshot,} and $\sigma^2_n{\bf e}_0$ becomes a deterministic noise term. Hence, the original model in (\ref{receive signal}) in the physical array domain ${\mathbb S}$, is converted into another model (\ref{vector}) in the difference co-array domain ${\cal D}$, and the DOA estimation can be applied to the data in (\ref{vector}) instead of (\ref{receive signal}). Each such technique, like co-array MUSIC \cite{2010Pal}, \cite{2011Pal}, actually amount to using a subvector ${\bf z}_{\cal U}$ of ${\bf z}$ to perform DOA estimation, where ${\cal U}=[-L_u,L_u]$ is the maximal central ULA segment of ${\cal D}$, and the number of uncorrelated sources that can be identified is $(|{\cal U}|-1)/2$.

				\subsection{Mutual Coupling}
				
				The received signal vector (\ref{receive signal}) assumes that the sensors do not interfere with each other. {In fact, the coupling between closely spaced sensors seriously affects the performance of DOA estimation. Coupling effect is added to the signal model. Hence, (\ref{receive signal}) can be  rewritten as}
				\begin{align}\label{mutual received vector}
					{\bf x}(t)={\bf CAs}(t)+{\bf n}(t),
				\end{align}
				where ${\bf C}$ is the $N\times N$ mutual coupling matrix.
				
				{Generally, the expression for ${\bf C}$ is rather complicated \cite{2016Liu-1}, \cite{2017Liu}. For the convenience of research, the entries of ${\bf C}$ are approximated by a B-band symmetric Toeplitz matrix mode  } \cite{1991Fried,2012Liao,2007Sellone,2000Svantesson,2009Ye}:
				\begin{align}\label{mutual coupling model}
					\langle {\bf  C}\rangle_{n_i,n_j}=\begin{cases}
						c_{|n_i-n_j|}, & \text{if $|n_i-n_j|\leq B$},\\
						0,& \text{otherwise},
					\end{cases}
				\end{align}
				where $n_i,n_j\in {\mathbb S}$ and the coupling coefficients satisfying \cite{1991Fried}
				\begin{align}\label{mutual coeff}
					\begin{cases}
						|c_{0}|=1>|c_1|>|c_2|>\cdots>|c_B|, \\
						|c_{i}|=0,i\geq B+1, \\
						|c_k/c_\ell|=\ell/k,k,l\in[1,B].
					\end{cases}
				\end{align}
				To evaluate the mutual coupling effect, the coupling leakage is usually used.

				\begin{definition}[Coupling Leakage]
					{For a given number of NLA, the total mutual coupling can be calculated by coupling leakage} \cite{2016Liu-1}, \cite{2019Zheng} as
					\begin{align}
						L_{c}=\frac{||{\bf C}-\mathrm{diag}({\bf C})||_F}{||{\bf C}||_F}
					\end{align} {where $||{\bf C}-\mathrm{diag}({\bf C})||_F$ is the energy of all the off-diagonal components, which characterizes the amount of mutual coupling. A large value of $L_{c}$ indicates that the severe coupling effect.}
				\end{definition}

				\section{The New Array Configuration}\label{New array}
				
				In this section, we will propose 3 types of new  array configurations under the new restriction (\ref{Our-condition}), including their uDOFs and the weight functions. The design of the newly proposed array adheres to the maximum basis standard, with $c$ elements on each side of the basis. By carefully configuring the element spacing on both sides of the basis, we achieve arrays with outstanding characteristics. The condition $N\geq18$ is given, since MRAs with small $N$ have been known in \cite{2002Van}.

					\subsection{New $(4r)$-Type Arrays}
					
					%
					When the base of the array is restricted to  be 0 modulo 4. Three classes of new $(4r)$-Type arrays can be obtained under the new restriction (\ref{Our-condition}), which are denoted as $(4r)$-Type 1 array, $(4r)$-Type 2 array and $(4r)$-Type 3 array respectively.

					Let $N=6r+k$ with $r\geq3, -2\leq k\leq 3$. Their spacing arrays and location arrays are expressed as follows.
					
					{\it$(4r)$-Type 1 Array:}
					\begin{align}\label{spacing-rep01}
						{\mathbb D}^{0,1}_{New}=&\{1,2^{r-2},3,(2r-1)^r,(4r)^{2r+k-1},\notag\\
						&(2r+1)^{r-1},2,2r-1,2^{r-1}\}.
					\end{align}
					\begin{align}\label{position-rep01}
						{\mathbb S}^{0,1}_{New}=&\{0,1,3,\ldots,2r-3,2r,4r-1,\ldots,2r^2+r,\notag\\
						& 2r^2+5r,2r^2+9r,\ldots,10r^2+(4k-3)r,\notag\\
						& 10r^2+(4k-1)r+1,\ldots,12r^2+(4k-4)r-1,\notag\\
						&12r^2+(4k-4)r+1,12r^2+(4k-2)r,\notag\\
						&12r^2+(4k-2)r+2,\ldots,12r^2+4kr-2\}.
					\end{align}
					
					{\it$(4r)$-Type 2 Array:}
					\begin{align}\label{spacing-rep02}
						{\mathbb D}^{0,2}_{New}=&\{1,2^{r-1},1,(2r-1)^r,(4r)^{2r+k-1},\notag\\
						&(2r+1)^r,2^{r-1}\}.
					\end{align}
					\begin{align}\label{position-rep02}
						{\mathbb S}^{0,2}_{New}=&\{0,1,3,\ldots,2r-1,2r,4r-1,\ldots,2r^2+r,\notag\\
						& 2r^2+5r,2r^2+9r,\ldots,10r^2+(4k-3)r,\notag\\
						& 10r^2+(4k-1)r+1,\ldots,12r^2+(4k-2)r,\notag\\
						&12r^2+(4k-2)r+2,\ldots,12r^2+4kr-2\}.
					\end{align}
					
					{\it$(4r)$-Type 3 Array:}
					\begin{align}\label{spacing-rep03}
						{\mathbb D}^{0,3}_{New}=&\{2^{r-1},1,(2r-1)^r,(4r)^{2r+k-1},\notag\\
						&(2r+1)^r,2^{r-1},1\}.
					\end{align}
					\begin{align}\label{position-rep03}
						{\mathbb S}^{0,3}_{New}=&\{0,2,4,\ldots,2r-1,2r,4r-2,\ldots,2r^2+r-1,\notag\\
						& 2r^2+5r-1,\ldots,10r^2+(4k-3)r-1,\notag\\
						& 10r^2+(4k-1)r,\ldots,12r^2+(4k-2)r-1,\notag\\
						&12r^2+(4k-2)r+1,\ldots,12r^2+4kr-3\notag\\
						& 12r^2+4kr-2\}.
					\end{align}
					
					{Clearly, the closed-form expressions of spacing arrays and position arrays are uniquely determined by parameters $r$ and $k$. In particular, the proposed arrays (\ref{position-rep01}), (\ref{position-rep02}) and (\ref{position-rep03}) generate a hole-free difference co-array, as indicated in Lemma \ref{hole-free}.}

					%
					
					\begin{lemma}\label{hole-free}
						The difference  co-arrays of the new $(4r)$-Type arrays (\ref{position-rep01}), (\ref{position-rep02}) and (\ref{position-rep03}) are hole-free ULAs, i.e., ${\cal D}_{New}=[-L,L]$ with $L=12r^2+4rk-2$.
					\end{lemma}
					\begin{proof}
						See Appendix \ref{appendix A}.
					\end{proof}
					

					%

					From Lemma \ref{hole-free}, we see that the difference co-array of the $N$-sensor new array is given by a consecutive set between $-(12r^2+4rk-2)$ and $12r^2+4rk-2$. Thus we obtain the following result.
					
					\begin{theorem}\label{thm-DOF}
						Let $N=6r+k$ with $r\geq3$ and $-2\leq k\leq 3$. The uDOFs for the $N$-sensor new $(4r)$-Type arrays (\ref{position-rep01}), (\ref{position-rep02}) and (\ref{position-rep03}):
						\begin{equation}
							\text{uDOFs}^{0,1}_{New}=\text{uDOFs}^{0,2}_{New}=\text{uDOFs}^{0,3}_{New}=\frac{2N^2}{3}-\frac{2k^2}{3}-3
						\end{equation}
						for any sensor number $N\geq18$. Thus its redundancy ratio is: $R^0_{New}<1.5$, and {$R^0_{New}=1.5$ when $N\rightarrow\infty$.}	
					\end{theorem}
					\begin{proof}
						From (\ref{position-rep01}), (\ref{position-rep02}), (\ref{position-rep03}) and  Lemma \ref{hole-free}, we know that
						\begin{align}\label{New-uDOFs} \text{uDOFs}^{0,1}_{New}&=\text{uDOFs}^{0,2}_{New}=\text{uDOFs}^{0,3}_{New}\notag\\
							&=2L+1=24r^2+8rk+3\\
							&=\frac{2N^2}{3}-\frac{2k^2}{3}-3.\notag
						\end{align}
						Thus the redundancy ratio is
						\begin{align}\label{Ratio}
							R_{New}^{0,1}&=R_{New}^{0,2}=R_{New}^{0,3}=\frac{N(N-1)}{2L}\notag\\&=\frac{36r^2+6r(2k-1)+k^2-k}{24r^2+8rk-4}\\
							&=\frac{3}{2}-\frac{6r-k^2+k-6}{24r^2+8rk-4}<1.5.\notag
						\end{align}	
						The last inequality is obtained by $N\geq 18$, i.e., $r=3, 0\leq k\leq 3$ and $r\geq 4,-2\leq k\leq3$.	
					\end{proof}
					
					


					The first three weight functions of the new $(4r)$-Type arrays can be also obtained as follows.

					\begin{theorem}\label{mutual}
						Let $N=6r+k$ with $r\geq3$ and $-2\leq k\leq 3$. For the $N$-sensor new arrays (\ref{position-rep01}), (\ref{position-rep02}) and (\ref{position-rep03}), their weight functions $\omega(a)$ at $a=1,2,3$ are
						
						{\it $(4r)$-Type 1 Array:}
						\begin{align}\label{new-mutua0l}
							\omega(1)=1,~\omega(2)=2r-2,~\omega(3)=2.
						\end{align}
						
						{\it $(4r)$-Type 2 Array:}
						\begin{align}\label{new-mutua02}
							\omega(1)=2,~\omega(2)=2r-2,~\omega(3)=2.
						\end{align}
						
						{\it $(4r)$-Type 3 Array:}
						\begin{align}\label{new-mutua03}
							\omega(1)=2,~\omega(2)=2r-2,~\omega(3)=2.
						\end{align}	
					\end{theorem}
					\begin{proof}
						To prove the results, we only need to count the numbers of continuous segment in ${\mathbb D}_{New}^{0,1}$ whose sums are 1, 2 and 3 respectively. Obviously, the spacing 1 only appears once as the first element of ${\mathbb D}_{New}^{0,1}$. The spacing 2 appears in
						\begin{equation}
							\{*,\underbrace{2,2,\ldots,2}\limits_{r-2~\text{times}},\cdots,\underbrace{2,}\limits_{1~\text{time}}*,\underbrace{2,2,\ldots,2}\limits_{r-1~\text{times}}\}.
						\end{equation}
						Thus $\omega(2)=r-2+1+r-1=2r-2$. The spacing 3 appears two times as sum in
						\begin{equation}
							\{\underbrace{1,2}\limits_{1~\text{time}},2^{r-3},\underbrace{3,}\limits_{1~\text{time}}\cdots\}.
						\end{equation}
						Thus $\omega(3)=2$. In this way, we have completed the proof of the weight function for the $(4r)$-Type 1 array. Moreover, we omit the proofs for the weight functions of the $(4r)$-Type 2 array and $(4r)$-Type 3 array, as the proof methods are similar.
					\end{proof}

					\begin{remark}
						Now we make some comparisons in  uDOFs and weight functions between our new $(4r)$-Type arrays and known arrays with the same type.
						\begin{enumerate}
							\item Compared with the $(4r)$-Type array-93 in \cite{1993Linebarger}, although our array $(4r)$-Type 1 has $4$ values less in uDOF for any $N\geq18$, it has a reduction in $\omega(1)$ and $\omega(2)$. Especially, it is the first class of array which achieves $R\leq 1.5$ and $\omega(1)=1$. Simultaneously, the $(4r)$-Type 2 and $(4r)$-Type 3 arrays also experience a loss of only $4$ uDOF, but they decrease $w(2)$. Note that the minor difference in uDOFs can be ignored as $N$ increases, especially incorporating the mutual coupling, for DOA estimation (See the details in the next section).
							\item Compared with the $(4r)$-Type array-10 in \cite{2010Dong}, our array has 2 values less in uDOFs for any $N\geq18$, but it also has a reduction in $\omega(1)$ and $\omega(2)$.
							\item Compared with  SNA in \cite{2016Liu-1} and MISC arrays in \cite{2019Zheng}, our array provides a significantly higher number of uDOFs than those of SNA and MISC array, in fact, $\text{uDOFs}^{0,1}_{New}=\text{uDOFs}^{0,2}_{New}=\text{uDOFs}^{0,3}_{New}\approx \frac{4}{3}\max\{\text{uDOFs}_{SNA},\text{uDOFs}_{MISC}\}$ when $N\rightarrow\infty $.
							Moreover, the mutual coupling of our array is also greatly reduced since $\omega^{0,1}_{New}(2)=\omega^{0,2}_{New}(2)=\omega^{0,3}_{New}(2)\approx\frac{N}{3}$ is far smaller than $\omega_{SNA}(2)\approx\omega_{MISC}(2)\approx\frac{N}{2}$.
							This advantage becomes evident as the number of sensors increases. This characteristic similarly applies to the $(4r+1)$-Type arrays presented later.
						\end{enumerate}
						
					\end{remark}

					\subsection{New $(4r+1)$-Type Arrays}
					When the base of the array is restricted to  be 1 modulo 4.
					Five classes of $(4r+1)$-Type arrays are obtained under the new restriction (\ref{Our-condition}), which are recorded as $(4r+1)$-Type 1 array, $(4r+1)$-Type 2 array, $(4r+1)$-Type 3 array, $(4r+1)$-Type 4 array and $(4r+1)$-Type 5 array respectively.
					
					Let $N=6r+k$ with $r\geq2, -1\leq k\leq 4$. Their spacing arrays and location arrays are expressed as follows.

					{\it $(4r+1)$-Type 1 Array: }
					\begin{align}\label{spacing-rep11}
						{\mathbb D}^{1,1}_{New}=&\{1^{r-1},(2r+1)^{r+1},(4r+1)^{2r+k-2},\notag\\
						&(2r)^{r},r-1,2,1^{r-1}\}.
					\end{align}
					\begin{align}\label{position-rep11}
						{\mathbb S}^{1,1}_{New}&=\{0,1,\ldots,r-1,3r,5r+1,\ldots,2r^2+4r,\notag\\
						& 2r^2+8r+1,\ldots,10r^2+(4k-2)r+k-2,\notag\\
						& 10r^2+4kr+k-2,\ldots,12r^2+(4k-2)r+k-2,\notag\\
						&12r^2+(4k-1)r+k-3,12r^2+(4k-1)r+k-1,\notag\\
						&12r^2+(4k-1)r+k,\ldots,12r^2+4kr+k-2\}.
					\end{align}

					{\it $(4r+1)$-Type 2 Array: }
					\begin{align}\label{spacing-rep12}
						{\mathbb D}^{1,2}_{New}=&\{1^{r-1},r-1,r+2,(2r+1)^{r},(4r+1)^{2r+k-2},\notag\\
						&(2r)^{r},r+1,1^{r-1}\}.
					\end{align}
					\begin{align}\label{position-rep12}
						{\mathbb S}^{1,2}_{New}&=\{0,1,\ldots,r-1,2r-2,3r,5r+1,\ldots,2r^2+4r,\notag\\
						& 2r^2+8r+1,\ldots,10r^2+(4k-2)r+k-2,\notag\\
						& 10r^2+4kr+k-2,\ldots,12r^2+(4k-2)r+k-2,\notag\\
						&12r^2+(4k-1)r+k-1,12r^2+(4k-1)r+k,\notag\\
						&\ldots,12r^2+4kr+k-2\}.
					\end{align}
					
					{\it $(4r+1)$-Type 3 Array: }
					\begin{align}\label{spacing-rep13}
						{\mathbb D}^{1,3}_{New}=&\{1^{r-1},r,r+1,(2r+1)^{r},(4r+1)^{2r+k-2},\notag\\
						&(2r)^{r},r+1,1^{r-1}\}.
					\end{align}
					\begin{align}\label{position-rep13}
						{\mathbb S}^{1,3}_{New}&=\{0,1,\ldots,r-1,2r-1,3r,5r+1,\ldots,2r^2+4r,\notag\\
						& 2r^2+8r+1,\ldots,10r^2+(4k-2)r+k-2,\notag\\
						& 10r^2+4kr+k-2,\ldots,12r^2+(4k-2)r+k-2,\notag\\
						&12r^2+(4k-1)r+k-1,12r^2+(4k-1)r+k,\notag\\
						&\ldots,12r^2+4kr+k-2\}.
					\end{align}
					
					{\it $(4r+1)$-Type 4 Array: }
					\begin{align}\label{spacing-rep14}
						{\mathbb D}^{1,4}_{New}=&\{1^{r-1},3,2r-2,(2r+1)^{r},(4r+1)^{2r+k-2},\notag\\
						&(2r)^{r},r+1,1^{r-1}\}.
					\end{align}
					\begin{align}\label{position-rep14}
						{\mathbb S}^{1,4}_{New}&=\{0,1,\ldots,r-1,r+2,3r,5r+1,\ldots,2r^2+4r,\notag\\
						& 2r^2+8r+1,\ldots,10r^2+(4k-2)r+k-2,\notag\\
						& 10r^2+4kr+k-2,\ldots,12r^2+(4k-2)r+k-2,\notag\\
						&12r^2+(4k-1)r+k-1,12r^2+(4k-1)r+k,\notag\\
						&\ldots,12r^2+4kr+k-2\}.
					\end{align}
					
					{\it $(4r+1)$-Type 5 Array: }
					\begin{align}\label{spacing-rep15}
						{\mathbb D}^{1,5}_{New}=&\{1^{r-1},(2r+1)^{r+1},(4r+1)^{2r+k-2},\notag\\
						&(2r)^{r},1,r,1^{r-1}\}.
					\end{align}
					\begin{align}\label{position-rep15}
						{\mathbb S}^{1,5}_{New}&=\{0,1,\ldots,r-1,3r,5r+1,\ldots,2r^2+4r,\notag\\
						& 2r^2+8r+1,\ldots,10r^2+(4k-2)r+k-2,\notag\\
						& 10r^2+4kr+k-2,\ldots,12r^2+(4k-2)r+k-2,\notag\\
						&12r^2+(4k-2)r+k-1,12r^2+(4k-1)r+k-1,\notag\\
						&12r^2+(4k-1)r+k,\ldots,12r^2+4kr+k-2\}.
					\end{align}
					
					
					The new $(4r+1)$-Type arrays also yield hole-free difference co-array, and their uDOFs and the first three weight functions can be also obtained as follows. Since their proofs are very similar as those of the $(4r)$-Type array, so we omit them.
					
					\begin{lemma}\label{hole-free1}
						The difference co-arrays of the new $(4r+1)$-Type arrays (\ref{position-rep11}), (\ref{position-rep12}), (\ref{position-rep13}), (\ref{position-rep14}) and (\ref{position-rep15}) are hole-free ULAs, i.e., ${\cal D}_{New}=[-L,L]$ with $L=12r^2+4rk+k-2$.
					\end{lemma}
					
					
					\begin{theorem}\label{thm-DOF1}
						Let $N=6r+k$ with $r\geq2$ and $-1\leq k\leq 4$. The uDOFs for the $N$-sensor new $(4r+1)$-Type arrays (\ref{position-rep11}), (\ref{position-rep12}), (\ref{position-rep13}), (\ref{position-rep14}) and (\ref{position-rep15}):
						\begin{align}
							&\text{uDOFs}^{1,1}_{New}=\text{uDOFs}^{1,2}_{New}=\text{uDOFs}^{1,3}_{New}=\text{uDOFs}^{1,4}_{New}\notag\\
							&=\text{uDOFs}^{1,5}_{New}=\frac{2N^2-2k^2+6k-9}{3}.
						\end{align}
						Thus their redundancy ratio are: $R^{1,1}_{New}=R^{1,2}_{New}<1.5$, and {$R^{1,1}_{New}=R^{1,2}_{New}=1.5$ when $N\rightarrow\infty$.}	
					\end{theorem}


						\begin{theorem}\label{mutual1}
							Let $N=6r+k$ with $r\geq2$ and $-1\leq k\leq 4$. For the $N$-sensor new arrays (\ref{position-rep11}), (\ref{position-rep12}), (\ref{position-rep13}), (\ref{position-rep14}) and (\ref{position-rep15}), their weight functions $\omega(m)$ at $m=1,2,3$ are
							
							{\it $(4r+1)$-Type 1 Array}:
							\begin{align}\label{new-mutual11}
								\omega(1)=2r-2,~\omega(2)=2r-4,~\omega(3)=2r-6.
							\end{align}	
							
							{\it $(4r+1)$-Type 2 Array}:
							\begin{align}\label{new-mutual12}
								\omega(1)=2r-2,~\omega(2)=2r-4,~\omega(3)=2r-6.
							\end{align}
							
							{\it $(4r+1)$-Type 3 Array}:
							\begin{align}\label{new-mutual13}
								\omega(1)=2r-2,~\omega(2)=2r-4,~\omega(3)=2r-6.
							\end{align}
							
							{\it $(4r+1)$-Type 4 Array}:
							\begin{align}\label{new-mutual14}
								\omega(1)=2r-2,~\omega(2)=2r-4,~\omega(3)=2r-5.
							\end{align}
							
							{\it $(4r+1)$-Type 5 Array}:
							\begin{align}\label{new-mutual15}
								\omega(1)=2r-1,~\omega(2)=2r-4,~\omega(3)=2r-6.
							\end{align}
						\end{theorem}
							

							\begin{remark}
								Compared with known arrays with $(4r+1)$-Type in Table \ref{table:redundancy thickapprox2}, our new arrays have at most 2 values less in uDOFs for any $N\geq18$, which may be ignored as $N$ increases. About the mutual coupling,
								\begin{enumerate}
									\item the new $(4r+1)$-Type 1 array, $(4r+1)$-Type 2 array and $(4r+1)$-Type 3 array have the lowest mutual coupling, since they provide a decrease in all of the first three weight functions.
									\item 
									Despite the fact that the $(4r)$-Type 4 array and the $(4r)$-Type 5 array show an increase in the values of $w(1)$ and $w(2)$ when compared to other recently proposed $(4r+1)$-Type arrays, they possess a lower weight function than the existing $(4r+1)$-Type arrays.
								\end{enumerate}


							\end{remark}

							\subsection{New $(4r+3)$-Type Arrays}
							
							When the base of array is restricted to 3 modulo 4. Two classes of $(4r+3)$-Type arrays are obtained under the new restriction (\ref{Our-condition}), which are recorded as $(4r+3)$-Type1 array and $(4r+3)$-Type2 array respectively.
							
							Let $N=6r+k$ with $r\geq3, 2\leq k\leq 7$. Their spacing arrays and location arrays are expressed as follows.
							
							{\it $(4r+3)$-Type 1 array:}
							\begin{align}\label{spacing-rep31}
								{\mathbb D}^{3,1}_{New}=&\{r+1,1^{r},(2r+2)^{r+1},(4r+3)^{2r+k-4},\notag\\
								&(2r+1)^{r},r+1,1^{r}\}.
							\end{align}
							\begin{align}\label{position-rep31}
								{\mathbb S}^{3,1}_{New}=&\{0,r+1,\ldots,2r+1,4r+3,\ldots,2r^2+6r+3,\notag\\
								& 2r^2+10r+6,\ldots,10r^2+(4k-4)r+3k-9,\notag\\
								& 10r^2+4kr-2r+3k-8,\ldots,12r^2+4rk-3r\notag\\
								&+3k-9,12r^2+4rk-2r+3k-8,\notag\\
								&12r^2+4rk-2r+3k-7,\ldots,\notag\\
								&12r^2+4kr-r+3k-8\}.
							\end{align}

							{\it $(4r+3)$-Type 2 array:}
							\begin{align}\label{spacing-rep32}
								{\mathbb D}^{3,2}_{New}=&\{1^{r},(2r+2)^{r+1},(4r+3)^{2r+k-4},\notag\\
								&(2r+1)^{r},r+1,1^{r},r+1\}.
							\end{align}
							\begin{align}\label{position-rep32}
								{\mathbb S}^{3,2}_{New}=&\{0,1,\ldots,r,3r+2,\ldots,2r^2+5r+2,\notag\\
								& 2r^2+9r+5,\ldots,10r^2+4rk-5r+3k-10,\notag\\
								& 10r^2+4kr+-3r+3k-10,\ldots,12r^2+4rk\notag\\
								&-4r+3k-10,12r^2+4rk-3r+3k-9,\notag\\
								&12r^2+4rk-3r+3k-8,\ldots,12r^2+4rk\notag\\
								&-2r+3k-9,12r^2+4kr-r+3k-8\}.
							\end{align}

							The similar conclusions on uDOFs and weight functions can be similarly obtained as follows.
							
							\begin{lemma}\label{hole-free2}
								The difference co-arrays of the new $(4r+3)$-Type arrays (\ref{position-rep31}) and (\ref{position-rep32}) are hole-free ULAs, i.e., ${\cal D}_{New}=[-L,L]$ with $L=12r^2+4rk-r+3k-8$.
							\end{lemma}
							
							
							\begin{theorem}\label{thm-DOF}
								Let $N=6r+k$ with $r\geq3$ and $2\leq k\leq 7$. The uDOFs for the $N$-sensor new $(4r+3)$-Type arrays (\ref{position-rep31}) and (\ref{position-rep32}) are:
								\begin{align*}
									\text{uDOFs}^{3,1}_{New}=\text{uDOFs}^{3,2}_{New}&=\frac{2N^2-N-2k^2+19k-45}{3}.
								\end{align*}
								Thus its redundancy ratio is: $R^{3,1}_{New}=R^{3,2}_{New}<1.5$, and {$R^{3,1}_{New}=R^{3,2}_{New}=1.5$ when $N\rightarrow\infty$.}	
							\end{theorem}

							
							\begin{theorem}\label{mutual2}
								Let $N=6r+k$ with $r\geq3$ and $2\leq k\leq 7$. For the $N$-sensor new arrays (\ref{position-rep31}) and (\ref{position-rep32}), their same weight functions $\omega(m)$ at $m=1,2,3$ are
								\begin{align}\label{new-mutual31}
									\omega(1)=2r,~\omega(2)=2r-2,~\omega(3)=2r-4.
								\end{align}	
							\end{theorem}
								
								
								\begin{remark}
									Compared with the known $(4r+3)$-Type-93 array, although our arrays have the same values for the first weight functions, but have a lower uDOFs. This is because the $(4r+3)$-Type-93 array may provide the largest uDOFs for all redundancy arrays with more than 8 sensors, since all known MRAs with more than 8 sensors coincide with this form, except for 13 sensors. We still list our new arrays with the same type here for completeness, perhaps they may provide alternatives in special situations such as high spatial resolution synthetic aperture radiometers \cite{2010Dong}.

								\end{remark}

								\section{Numerical Examples} \label{numerical examples}
								
								In this section, {in order to illustrate the advantages of the proposed array, the weight function, coupling leakage, mutual coupling matrix and DOA estimation performance of the array are compared through specific numerical examples.}
								Due to limitations in the length of the article, six types of LRAs are chosen for simulation comparisons: the second-order super nested array (SNA) \cite{2016Liu-1}, the MISC array \cite{2019Zheng}, the $(4r+3)$-Type-93 array in \cite{1993Linebarger}, the $(4r)$-Type-93 array in \cite{1993Linebarger}, the $(4r)$-Type-10 array in  \cite{2010Dong}, and our new $(4r)$-Type 1 array  defined in (\ref{spacing-rep01}).
								{Co-array MUSIC algorithm is used to execute DOA estimation so that the information on DCA can be fully utilized.}
								To evaluate the DOA estimation performance of the sparse arrays, the root-mean-square error (RMSE) of the estimated normalized DOAs is shown as:
								\begin{equation}\label{eq:ss21}
									\text{RMSE}=\sqrt{\frac{1}{QK}\sum_{q=1}^{Q}\sum_{k=1}^{K}(\tilde{\theta}_{k}^{q}-\bar{\theta}_{k})^{2}}
								\end{equation}
								where $Q$ is the number of independent trials, and $\tilde{\theta}_{k}^{q}$ is the estimate of $\bar{\theta}_{k}$ for the $k$th trial. {Moreover, we adapt the mutual coupling model in \cite{2016Liu-1}, which is characterized by $c_{1}=0.3e^{j\pi/3}, B=100$ and $c_{\ell}=c_{1}e^{-j(\ell-1)\pi/8}/l$ for $2\leq \ell\leq B$.}

								\subsection{{Weight Functions, Coupling Leakage and Mutual Coupling Matrices}}
								{In this subsection, we compare
									their weight functions, coupling leakage and the mutual coupling matrices.}
								Three different cases are considered  where the number of sensors is $18$, $23$ and $36$, respectively.

								{For a given array element number $N$, we provides a summary of the weight function and the mutual coupling leakage $L_c$ in Table \ref{table:weightandcoupling},} where only the weight functions $\omega(1)$, $\omega(2)$ and $\omega(3)$ are shown, since they provide major impact on the mutual coupling effect. It is shown that the $(4r+3)$-Type-93 array exhibits the largest $L_c$, due to the dense ULA in its configurations, which results in severe mutual coupling effect.
								The rest five LRAs have relatively small $L_c$ because they all greatly reduce the numbers of small inter-sensor spacings. Generally, the value $L_c$ decreases as the sensor number increases in all the $(4r)$-Type arrays, and their $L_c$s are lower than those of SNA and MISC array when $N=36$.
								Especially, our new $(4r)$-Type array attains the smallest $L_c$ compared with other LRAs, which implies that our array is the best to resist mutual coupling effect.
								
								\begin{table*}[htbp]
									\caption{THE WEIGHT FUNCTIONS AND MUTUAL COUPLING LEAKAGE FOR SIX KINDS OF LRAs\label{table:weightandcoupling}}
									\centering
									\small{
										\begin{tabular}{ccccccccc} 
											\hline
											Array  & $(4r+3)$-Type-93 & $(4r)$-Type-93  & $(4r)$-Type-10  & SNA  & MISC & New $(4r)$-Type 1\\
											\hline
											18 sensors & &  &    &      & &  &  & \\\hline
											$w(1)$ & 4 & 2 & 3    &  1  & 1 & 1 \\
											$w(2)$ & 2 & 5 & 5    &  8  & 6 & 4 \\
											$w(3)$ & 1 & 2 & 2    &  1  &  1 & 2 \\
											$L_{c}$ & 0.2209& 0.1993 & 0.2200     &  0.1979   & 0.1802 & 0.1659 \\\hline
											23 sensors & & &       & &  &  &  \\\hline
											$w(1)$ &6& 2& 3    &  1    & 1 & 1 \\
											$w(2)$ &4&7 & 7    &  10  & 8 & 6 \\
											$w(3)$ &2& 2& 2    &  1    & 1 & 2 \\
											$L_{c}$ &0.2383& 0.1902& 0.2068     &  0.1910  & 0.1762 & 0.1612 \\\hline
											36 sensors &&  &      &  &   &  &  \\\hline
											$w(1)$ &10& 2& 3    &  2  & 1 & 1 \\
											$w(2)$ &8& 11 & 11    &  15  & 16 & 10 \\
											$w(3)$ &6&2 & 2    &  4  & 1 & 2 \\
											$L_{c}$ &0.2499&0.1737  & 0.1852     &  0.1988  & 0.1856 & 0.1524 \\\hline
									\end{tabular}}
								\end{table*}

								Fig. \ref{fig:weightfunction} and
								Fig. \ref{fig:couplingmatrices} give visual representations of the weight functions, or the magnitudes of the mutual coupling matrices for the six LRAs, respectively. In Fig. \ref{fig:weightfunction}, the heights of these line segments represent the sizes of weight functions $\omega(l)$ with $-20\leq l\leq 20$, which is symmetrical about $l=0$. It is easy to see that our new array has the lowest line segments at $l=1,2$, compared with other LRAs. In Fig. \ref{fig:couplingmatrices}, the color of blocks represents the energy in the corresponding entry, where the less color implies less energy. Thus we want the light-colored blocks to be as little as possible. From Fig. \ref{fig:couplingmatrices} we can still demonstrate the superiority of our new array.

								\begin{figure*}[htbp]
									\centering
									\begin{minipage}[t]{0.31\linewidth}
										\centering
										\includegraphics[width=1.8in]{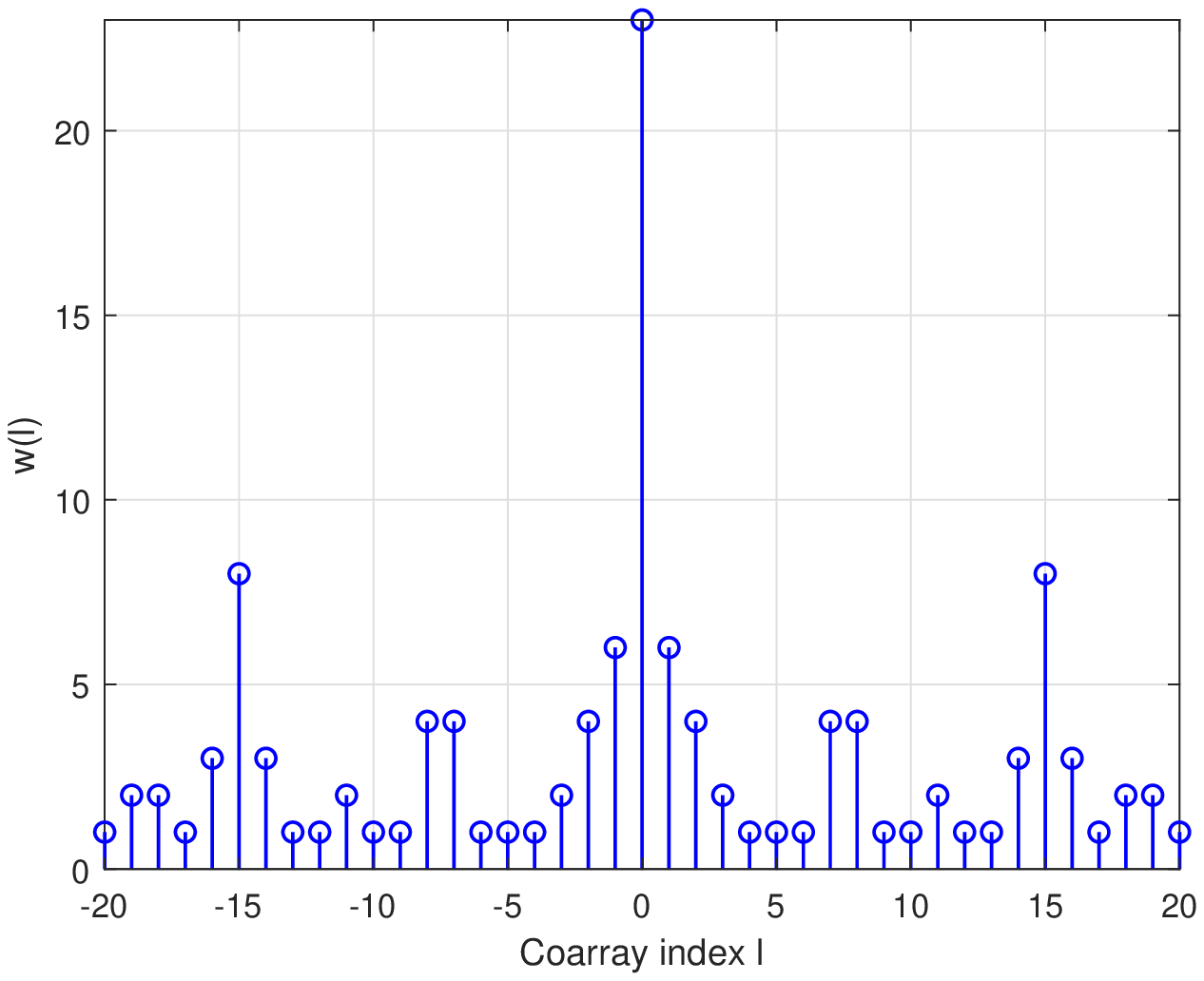}
										
										{(a)}
										\label{fig:side:b}
									\end{minipage}
									\begin{minipage}[t]{0.31\linewidth}
										\centering
										\includegraphics[width=1.8in]{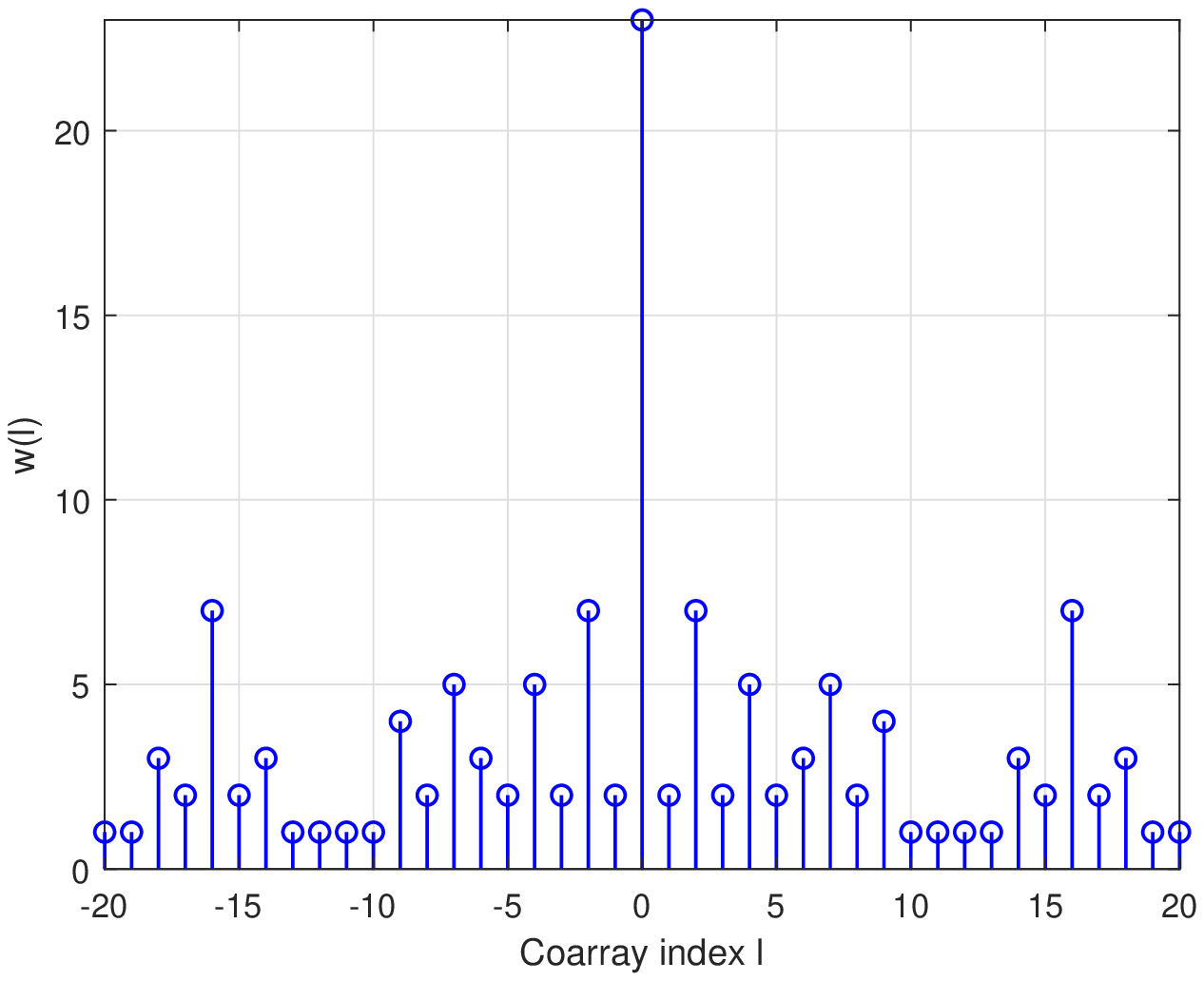}
										
										{(b)}
										\label{fig:side:b}
									\end{minipage}%
									\begin{minipage}[t]{0.31\linewidth}
										\centering
										\includegraphics[width=1.8in]{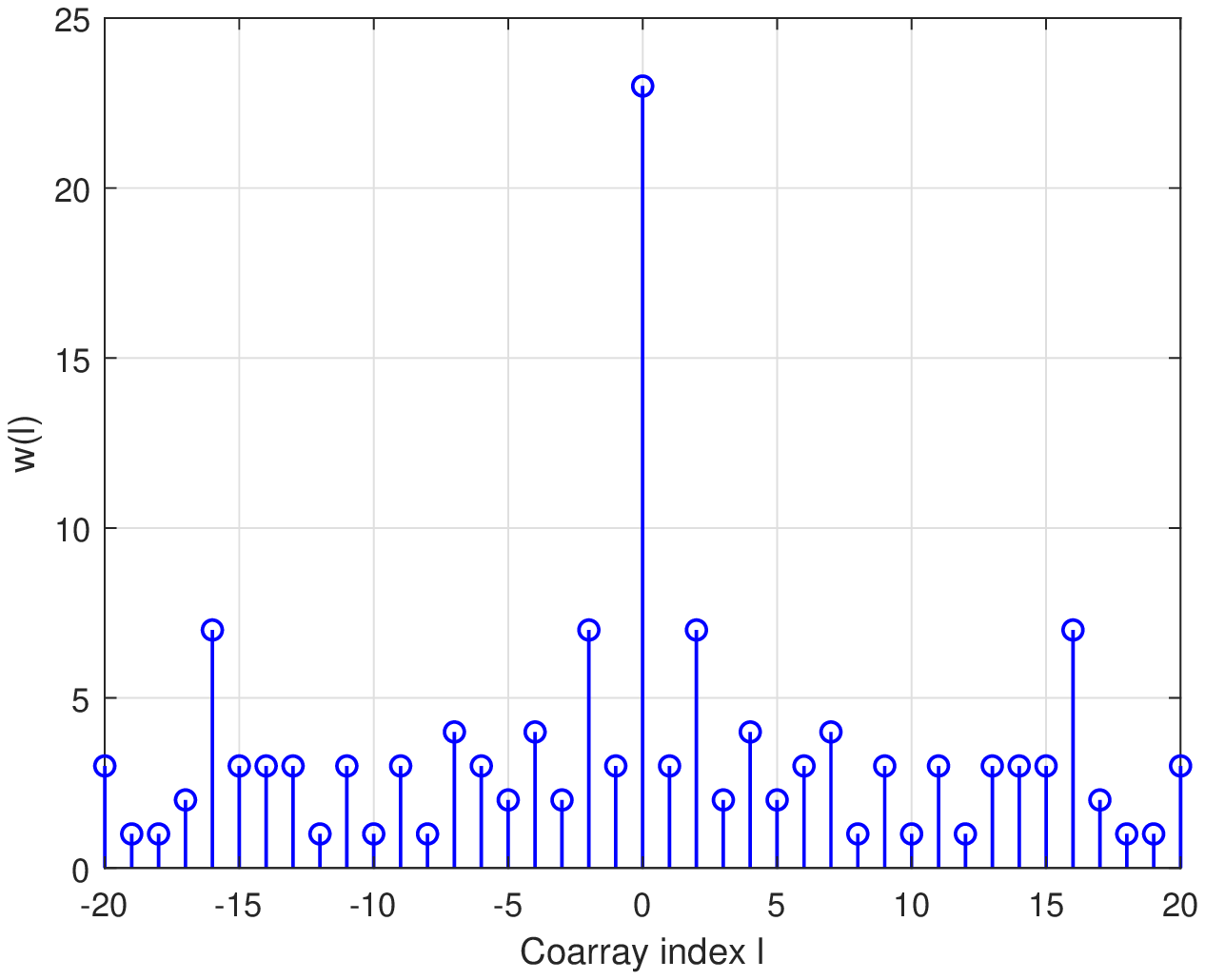}
										
										{(c)}
										\label{fig:side:b}
									\end{minipage}
									
									\begin{minipage}[t]{0.31\linewidth}
										\centering
										\includegraphics[width=1.8in]{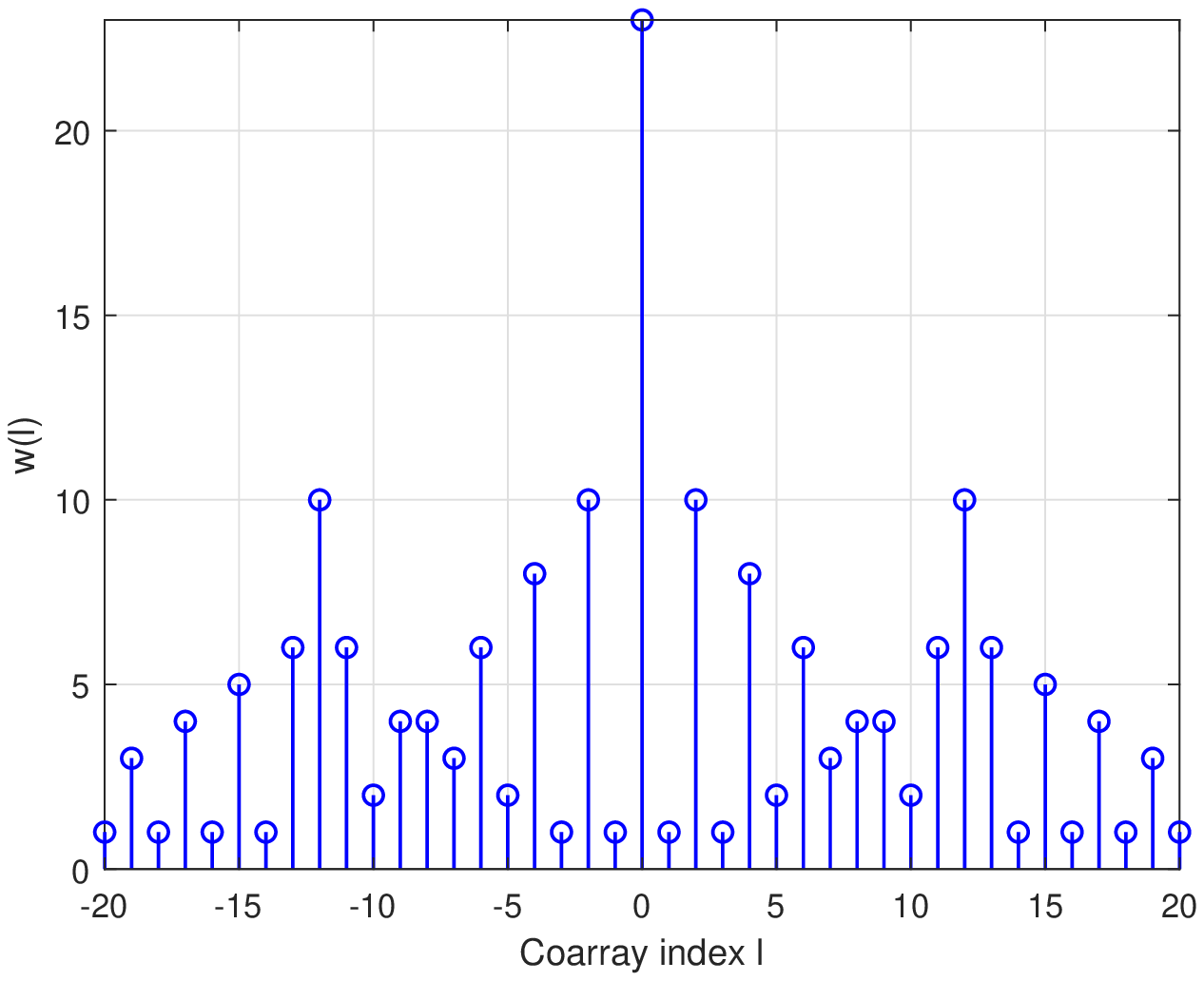}
										
										{(d)}
										\label{fig:side:b}
									\end{minipage}
									\begin{minipage}[t]{0.31\linewidth}
										\centering
										\includegraphics[width=1.8in]{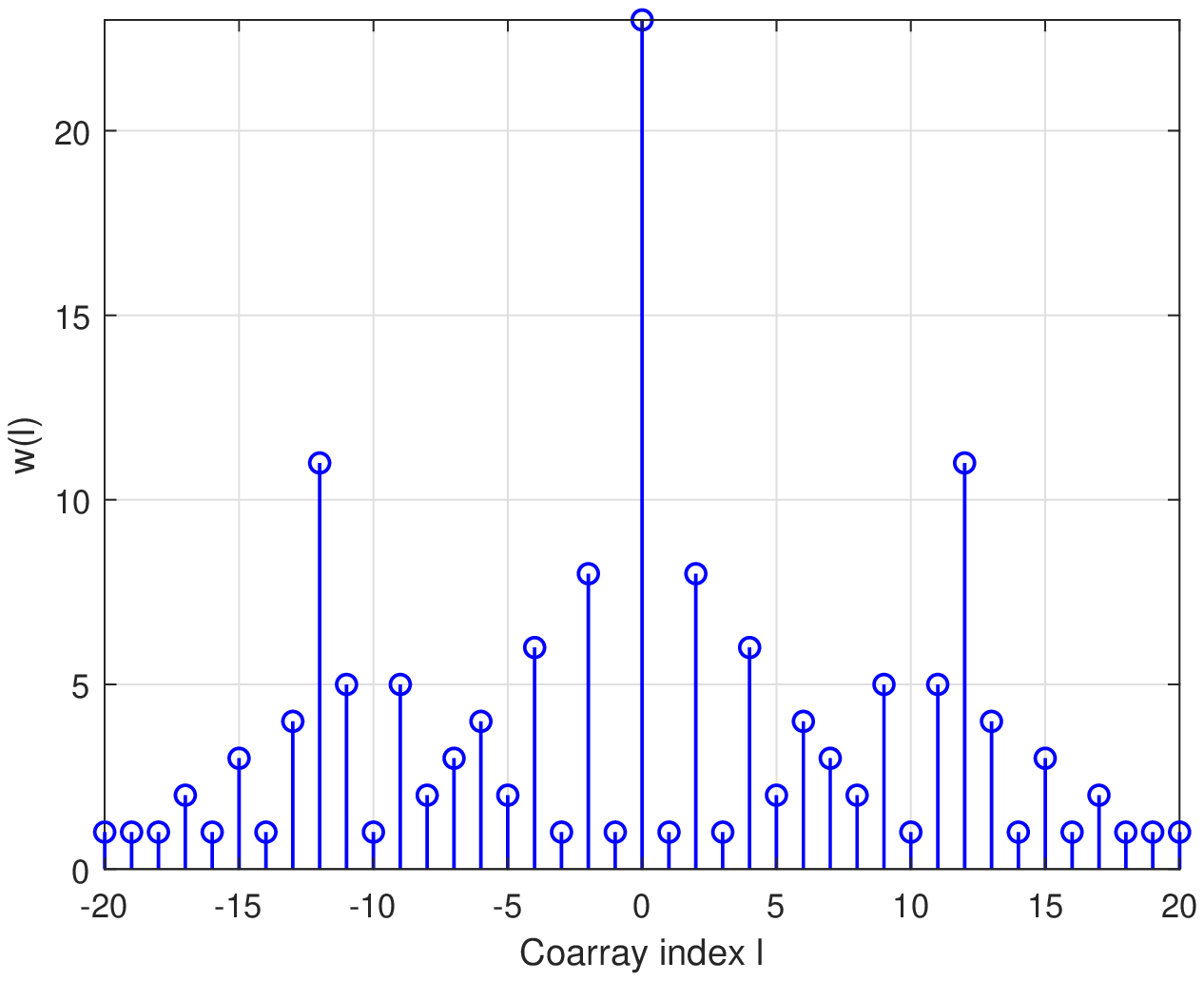}
										
										{(e)}
										\label{fig:side:b}
									\end{minipage}
									\begin{minipage}[t]{0.31\linewidth}
										\centering
										\includegraphics[width=1.8in]{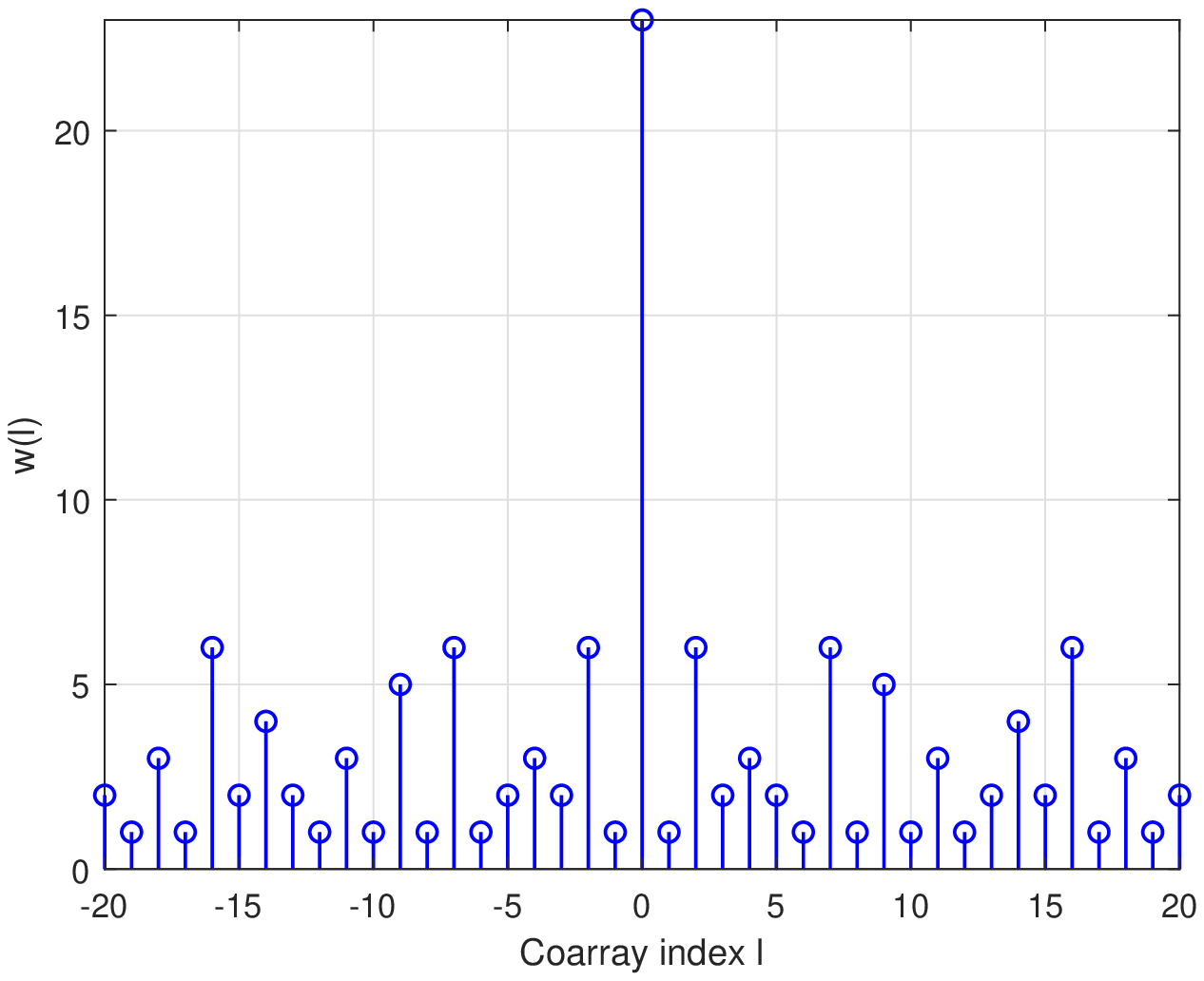}
										
										{(f)}
										\label{fig:side:c}
									\end{minipage}
									
									\caption{The weight functions for six kinds of $23$-element LRAs. (a) $(4r+3)$-Type-93. (b) $(4r)$-Type-93. (c) $(4r)$-Type-10. (d) SNA. (e) MISC. (f) New $(4r)$-Type 1.}
									\label{fig:weightfunction}
								\end{figure*}

								\begin{figure*}[htbp]
									\centering
									\begin{minipage}[t]{0.31\linewidth}
										\centering
										\includegraphics[width=1.8in]{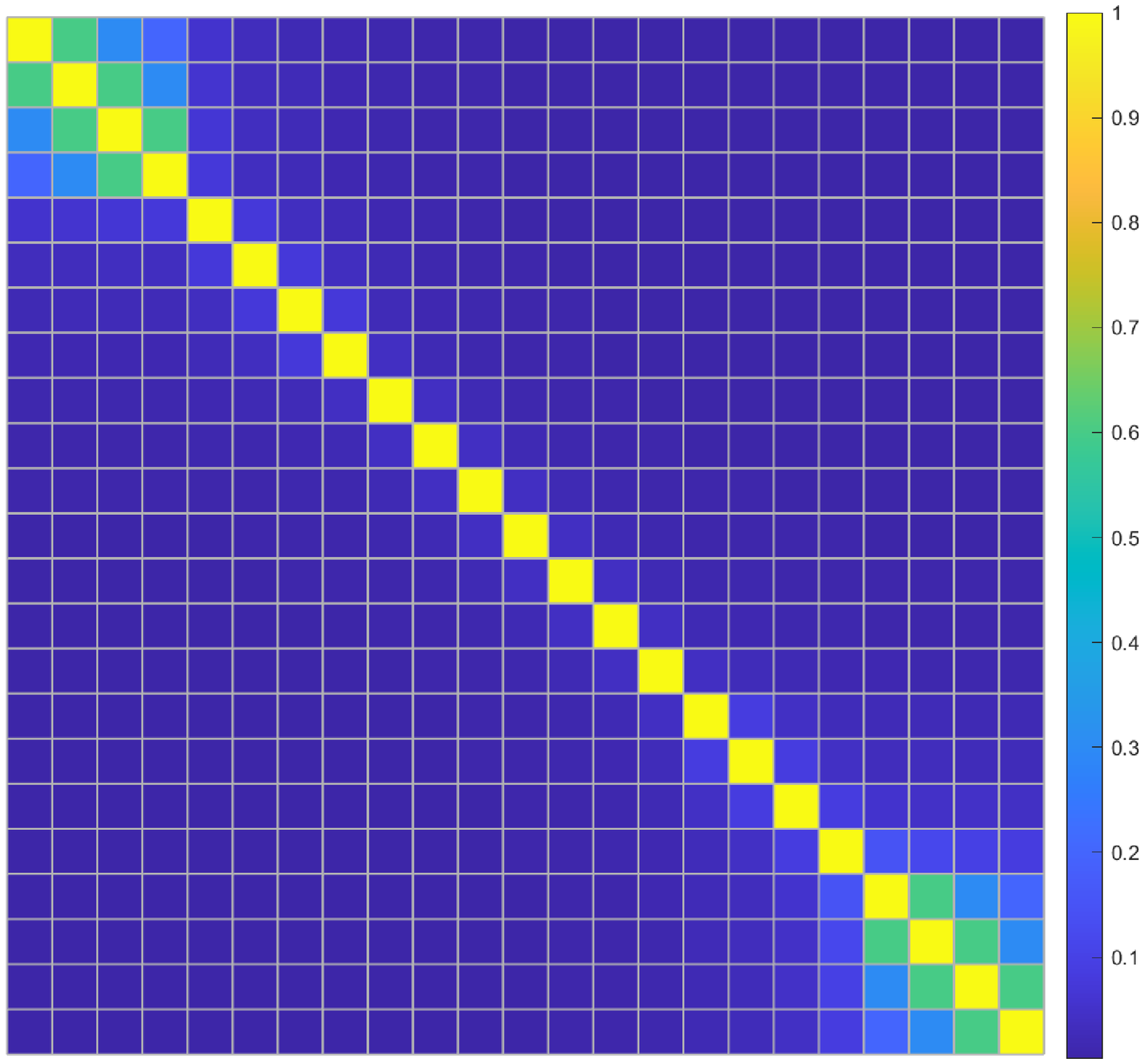}
										
										{(a)}
										\label{fig:side:b}
									\end{minipage}
									\begin{minipage}[t]{0.31\linewidth}
										\centering
										\includegraphics[width=1.8in]{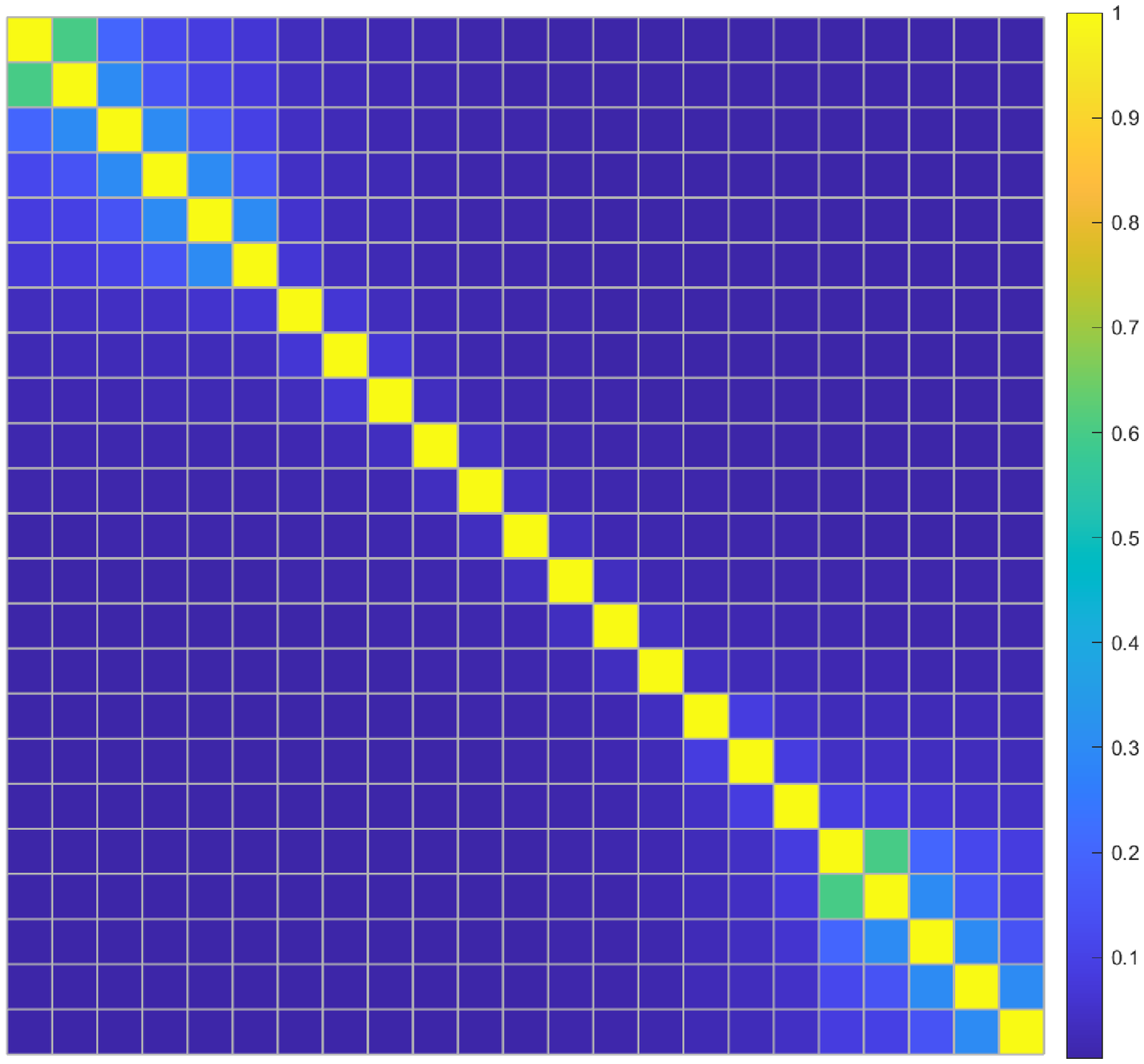}
										
										{(b)}
										\label{fig:side:b}
									\end{minipage}
									\begin{minipage}[t]{0.31\linewidth}
										\centering
										\includegraphics[width=1.8in]{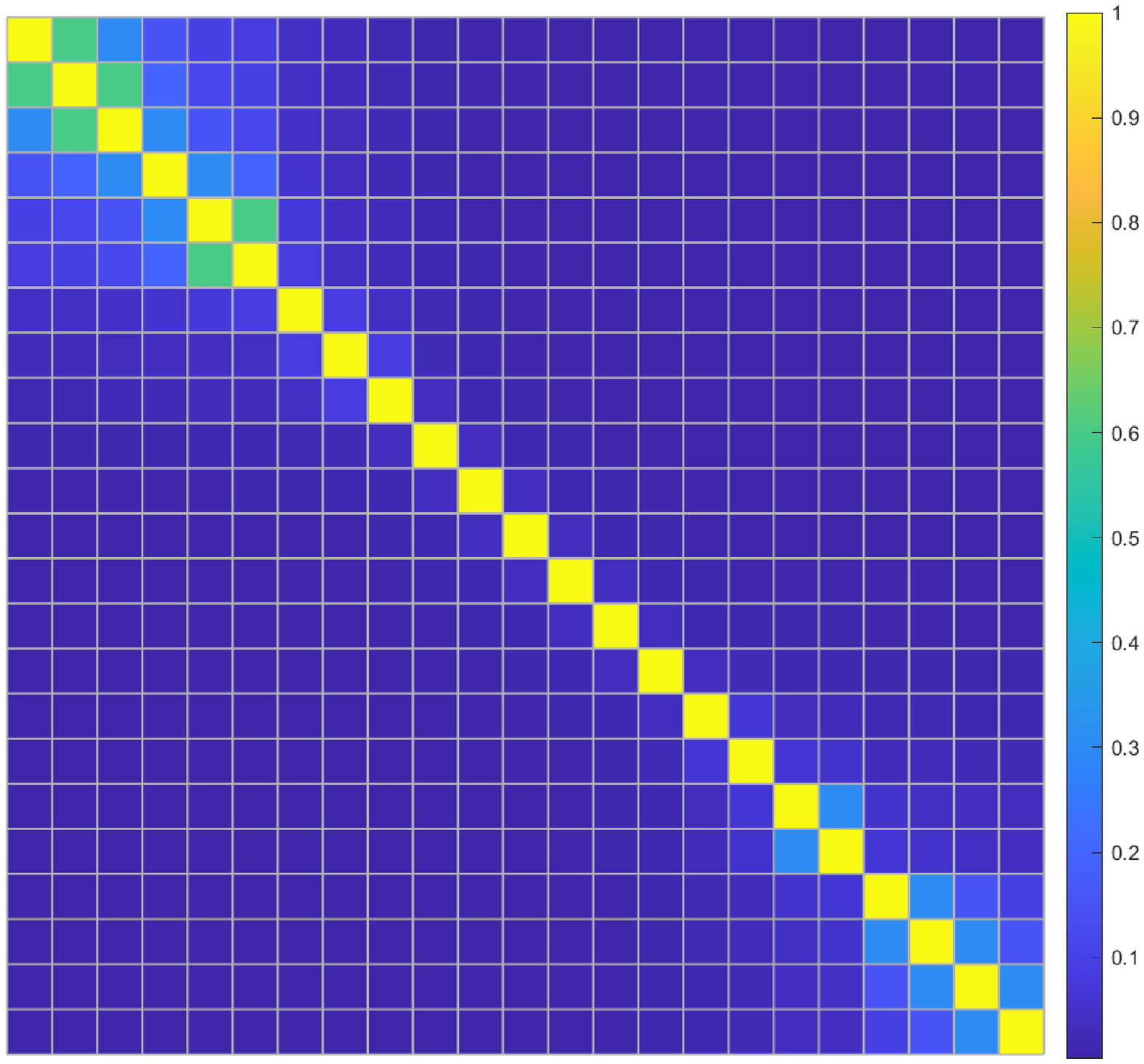}
										
										{(c)}
										\label{fig:side:a}
									\end{minipage}%
									
									\begin{minipage}[t]{0.31\linewidth}
										\centering
										\includegraphics[width=1.8in]{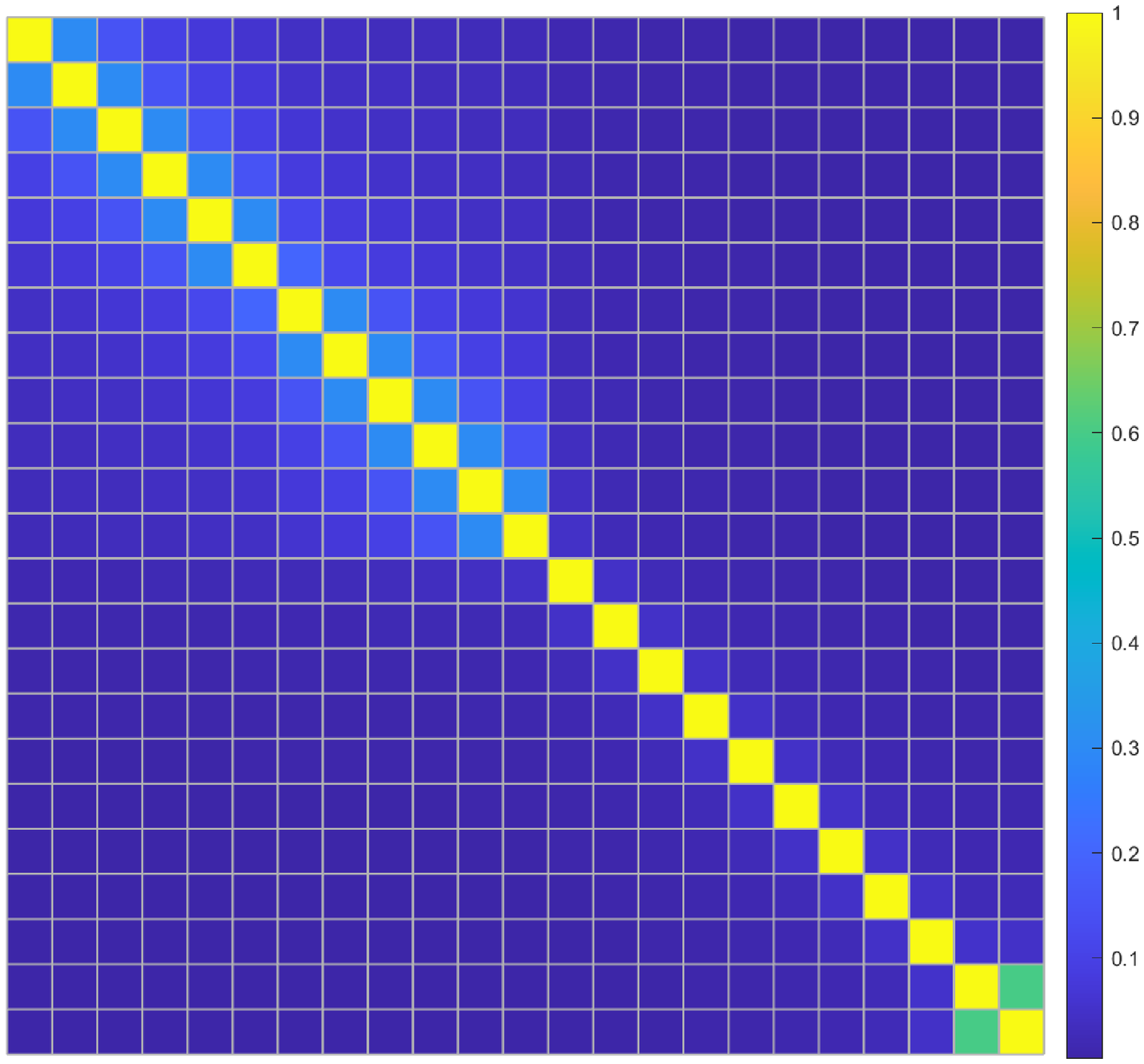}
										
										{(d)}
										\label{fig:side:b}
									\end{minipage}
									\begin{minipage}[t]{0.31\linewidth}
										\centering
										\includegraphics[width=1.8in]{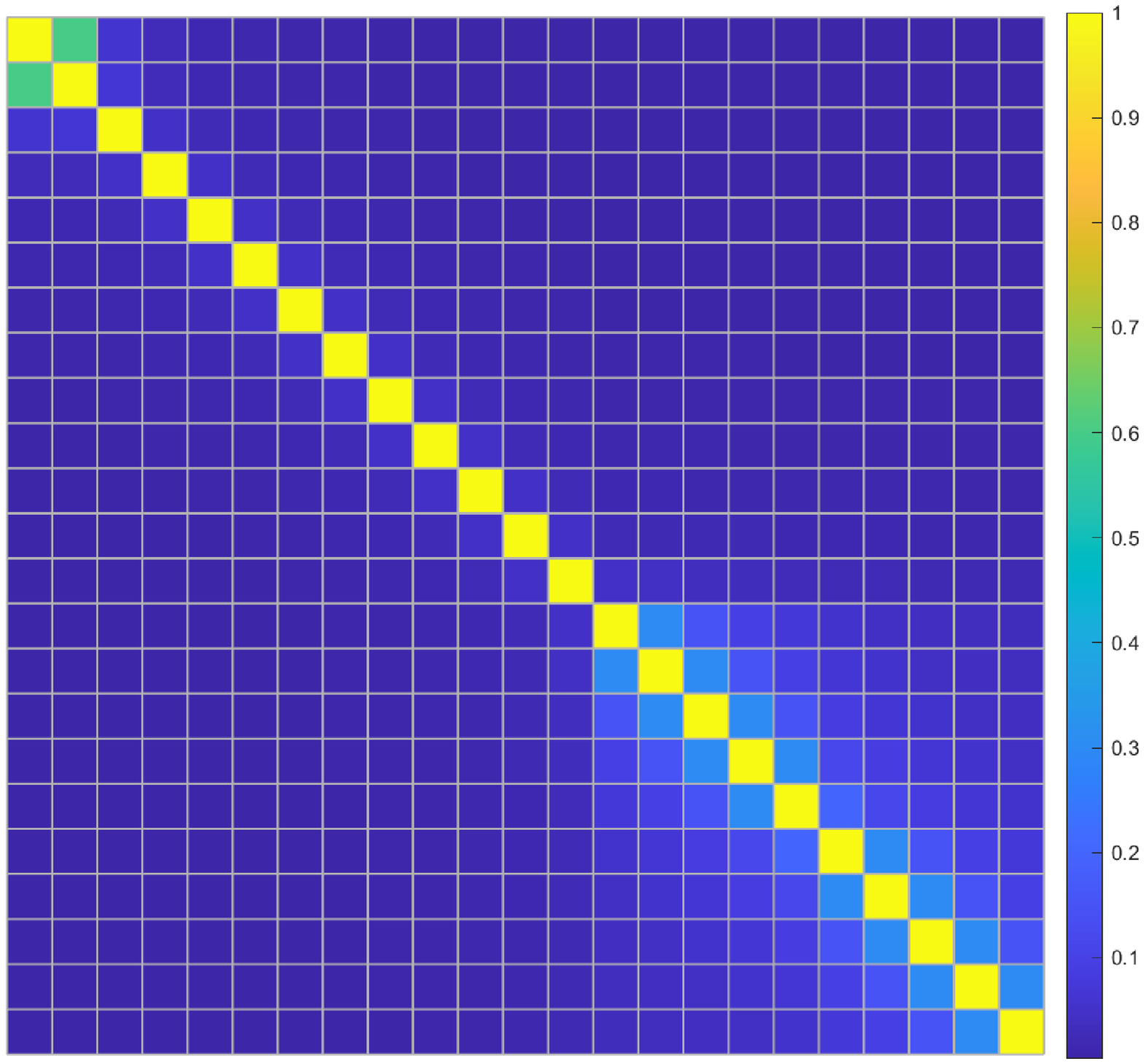}
										
										{(e)}
										\label{fig:side:b}
									\end{minipage}
									\begin{minipage}[t]{0.31\linewidth}
										\centering
										\includegraphics[width=1.8in]{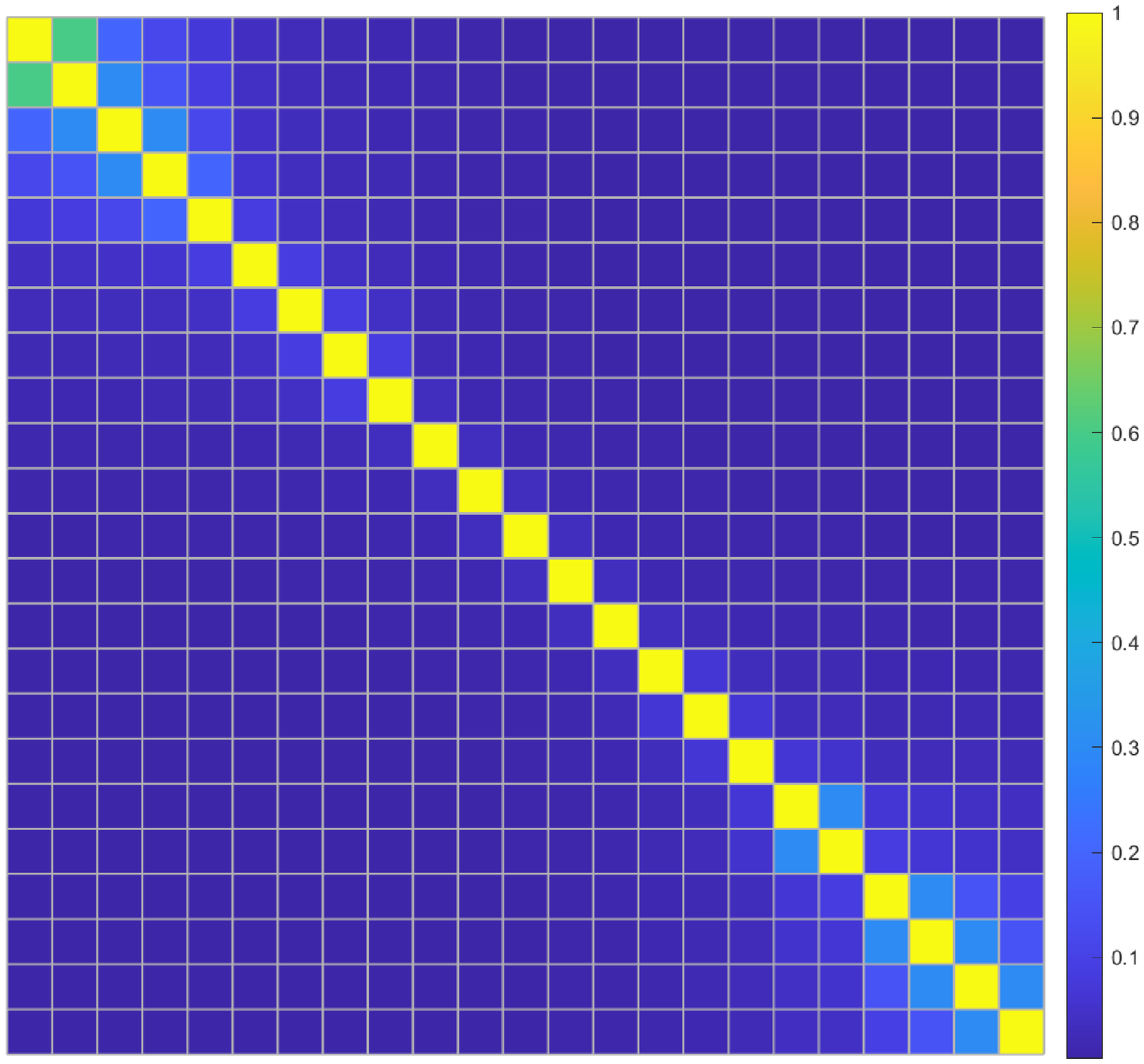}
										
										{(f)}
										\label{fig:side:c}
									\end{minipage}
									
									\caption{The magnitudes of the mutual coupling matrices for six kinds of $23$-element LRAs. (a) $(4r+3)$-Type-93. (b) $(4r)$-Type-93. (c) $(4r)$-Type-10. (d) SNA. (e) MISC. (f) New $(4r)$-Type 1.}
									\label{fig:couplingmatrices}
								\end{figure*}

								\subsection{DOA Estimation in the Presence of Mutual Coupling}
								
								In this subsection, we compare the DOA estimation performance of these six LRAs in the presence of mutual coupling. The same sensor number 23 is used for all arrays, and the sensor position sets for these arrays are given in Table \ref{table:sensor position}.
								
								\begin{table*}[htbp]
									\centering
										\caption{SENSOR POSITION SUMMARY OF $23$ ELEMENTS LRAs\label{table:sensor position}}
										\begin{tabular}{|c|c|}\hline
											& Sensor Position \cr\hline
											SNA & \{0,2,4,6,8,10,13,15,17,19,21,23,35,47,59,71,83,95,107,119,131,142,143\} \cr\hline
											MISC & \{0,1,10,22,34,46,58,70,82,94,106,118,130,142,144,146,148,150,153,155,157,159,161\} \cr\hline
											$(4r+3)$-Type-93 & \{0,1,2,3,11,19,27,35,50,65,80,95,110,125,140,155,162,169,176,180,181,182,183\} \cr\hline
											$(4r)$-Type-93 & \{0,1,3,5,7,9,18,27,36,52,68,84,100,116,132,148,155,162,169,170,172,174,176\} \cr\hline
											$(4r)$-Type-10 & \{0,1,2,4,6,7,14,21,28,44,60,76,92,108,124,140,149,158,160,169,171,173,175\} \cr\hline
											New $(4r)$-Type 1 & \{ 0,1,3,5,8,15,22,29,36,52,68,84,100,116,132,141,150,159,161,168,170,172,174\}  \cr\hline
										\end{tabular}
								\end{table*}

								
								\subsubsection{MUSIC Spectra}

								Fig. $\ref{fig:MUSICspectracoupling}$ shows the MUSIC spectra for the six kinds of $23$-element arrays when $K=40$ sources are uniformly located at $\bar{\theta}_k=-0.45+0.9(k-1)/39,1\leq k\leq 40.$ The SNR is fixed at $0$ dB and the number of snapshots is set as $T=1000.$ Observe from Fig. \ref{fig:MUSICspectracoupling}, only the $(4r)$-Type-93 array and our new array $(4r)$-Type are capable of detecting all $40$ sources clearly, while the other arrays (with false peaks or missing peaks) are not. Especially, our new array has higher peaks than the $(4r)$-Type-93 array, indicating the higher resolution. Since the numbers of uDOFs of these arrays are all higher than 40, our new array is more effective than the remaining arrays against strong mutual coupling.
								
								\begin{figure*}[htbp]
									\centering
									\begin{minipage}[t]{0.31\linewidth}
										\centering
										\includegraphics[width=1.8in]{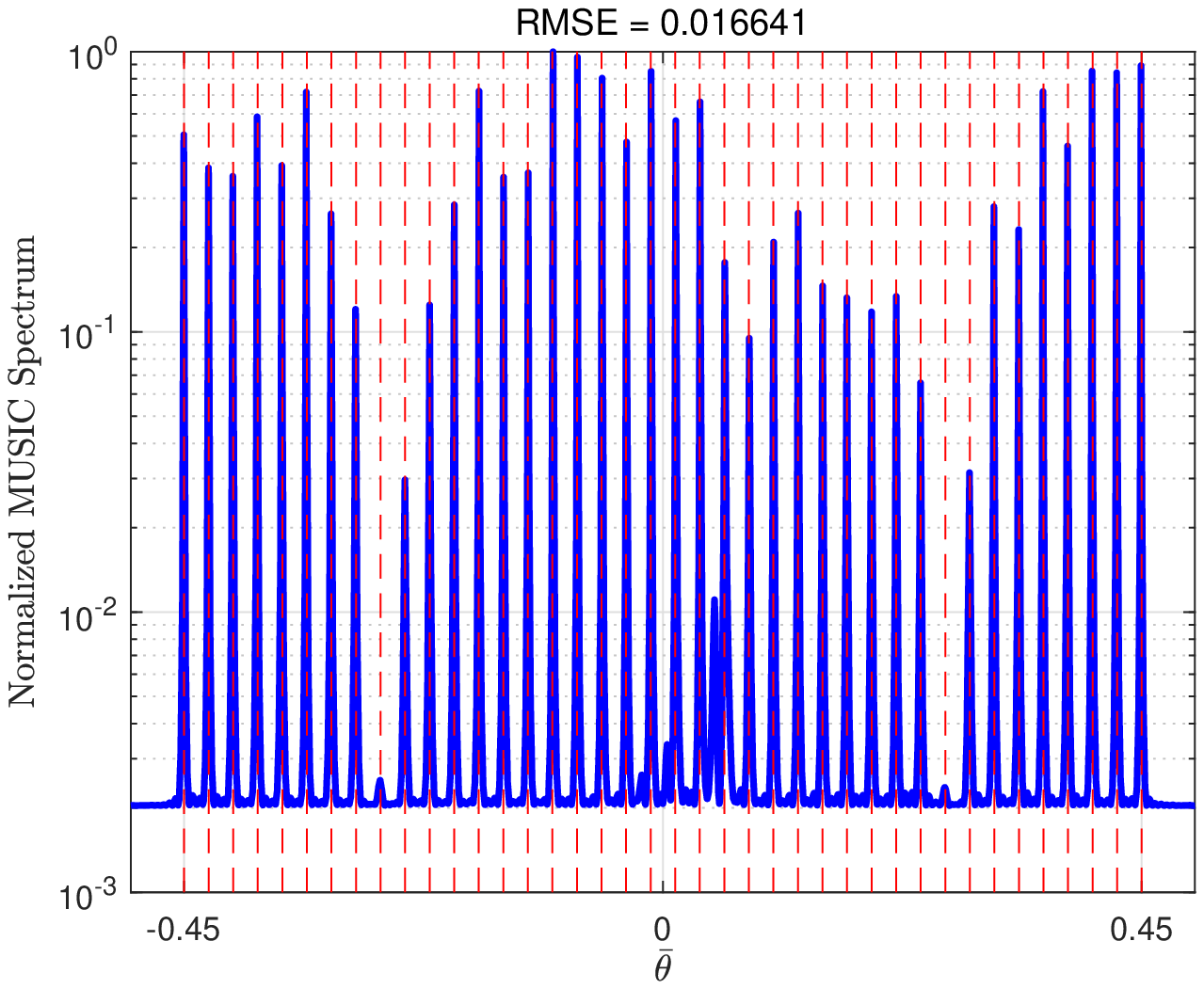}
										
										{(a)}
										\label{fig:side:b}
									\end{minipage}
									\begin{minipage}[t]{0.31\linewidth}
										\centering
										\includegraphics[width=1.8in]{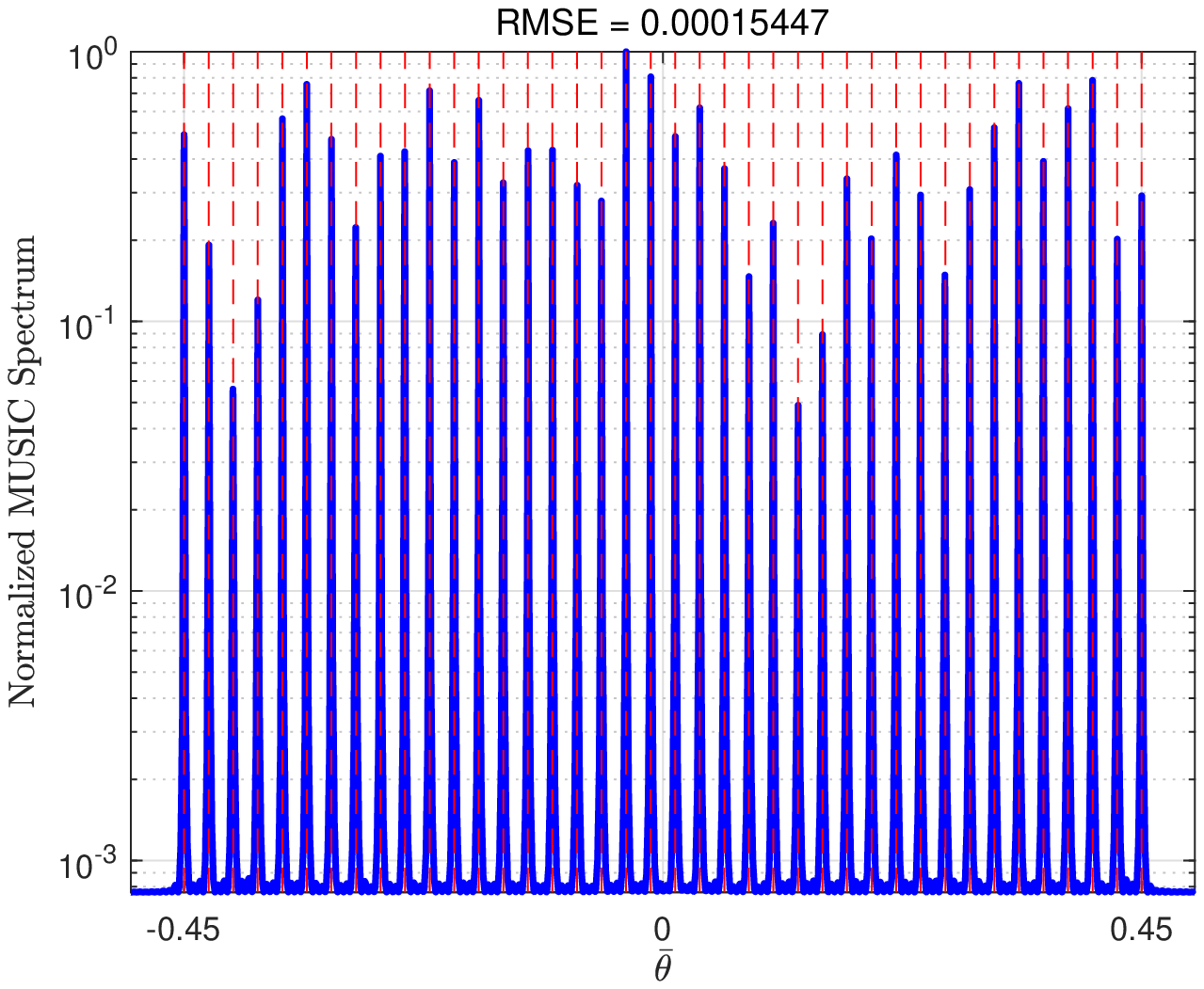}
										
										{(b)}
										\label{fig:side:c}
									\end{minipage}
									\begin{minipage}[t]{0.31\linewidth}
										\centering
										\includegraphics[width=1.8in]{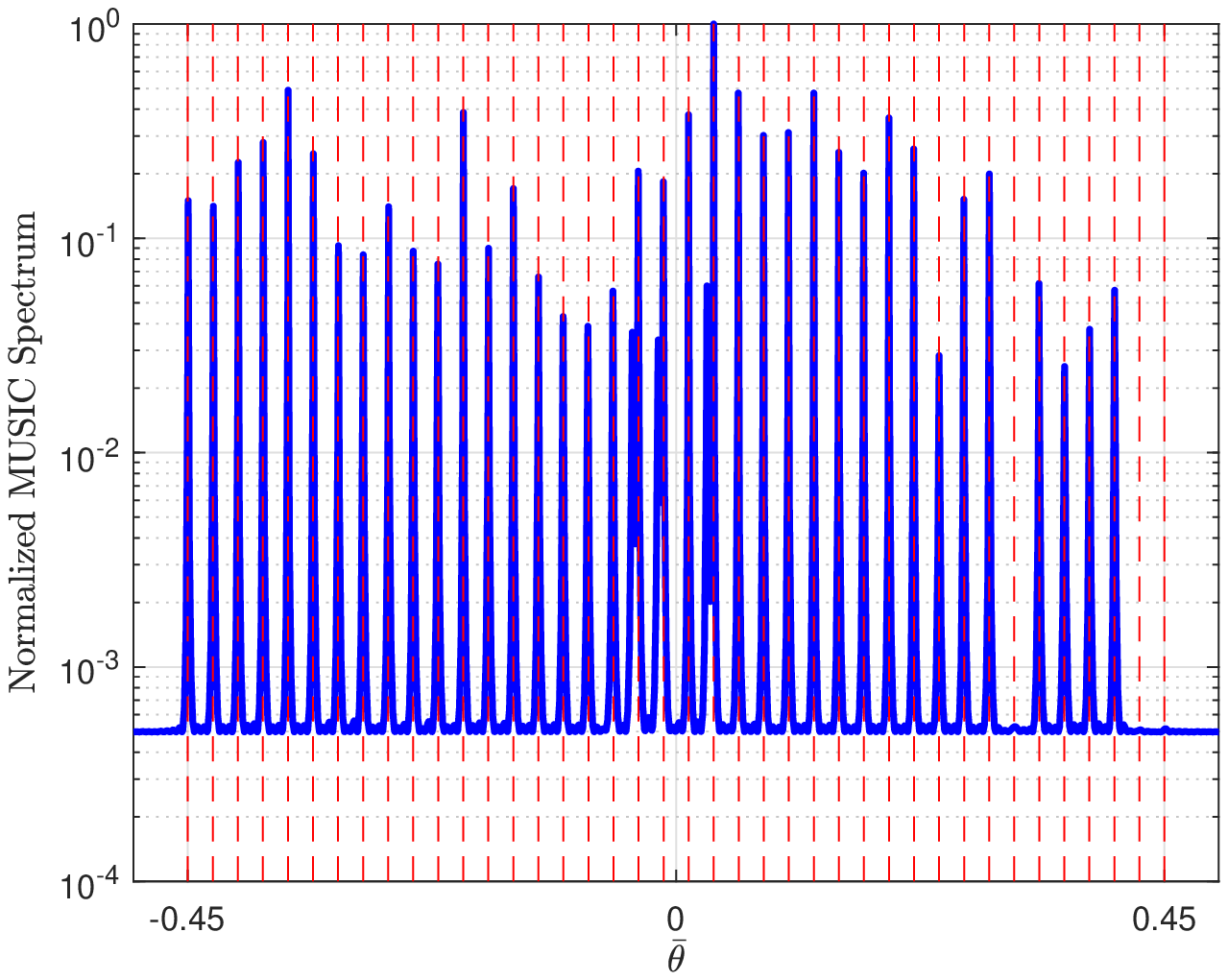}
										
										{(c)}
										\label{fig:side:b}
									\end{minipage}
									
									\begin{minipage}[t]{0.31\linewidth}
										\centering
										\includegraphics[width=1.8in]{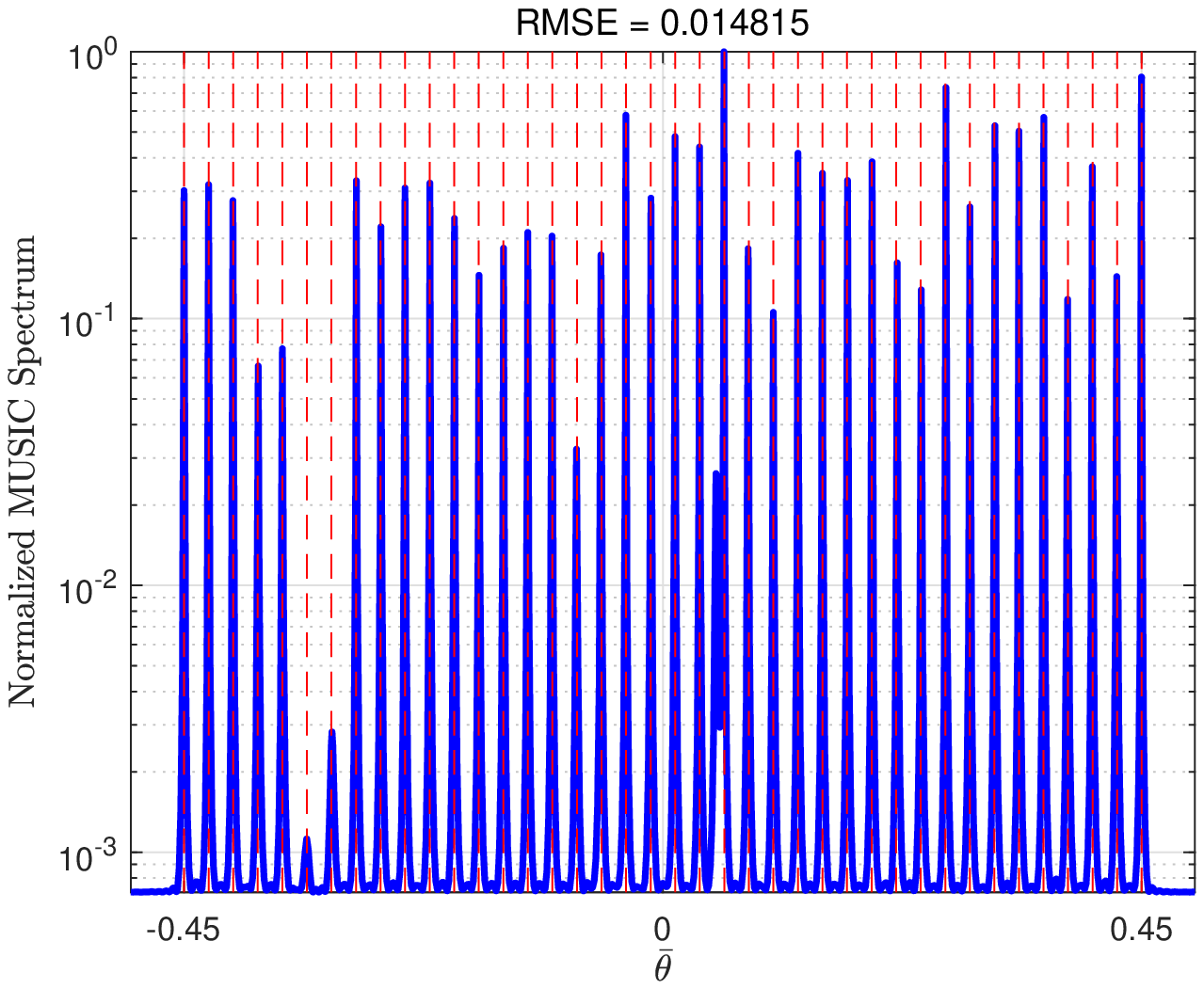}
										
										{(d)}
										\label{fig:side:c}
									\end{minipage}
									\begin{minipage}[t]{0.31\linewidth}
										\centering
										\includegraphics[width=1.8in]{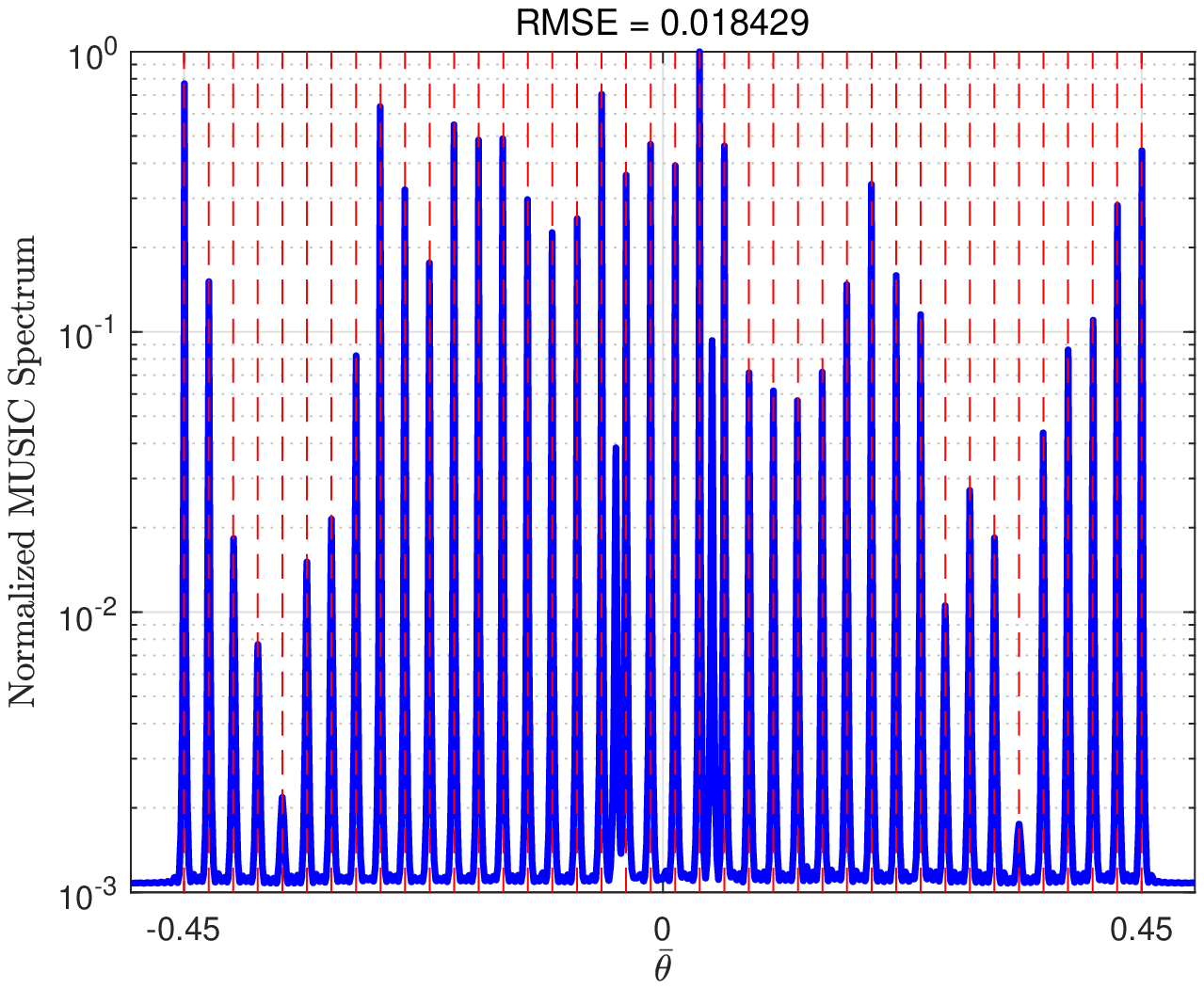}
										
										{(e)}
										\label{fig:side:a}
									\end{minipage}%
									\begin{minipage}[t]{0.31\linewidth}
										\centering
										\includegraphics[width=1.8in]{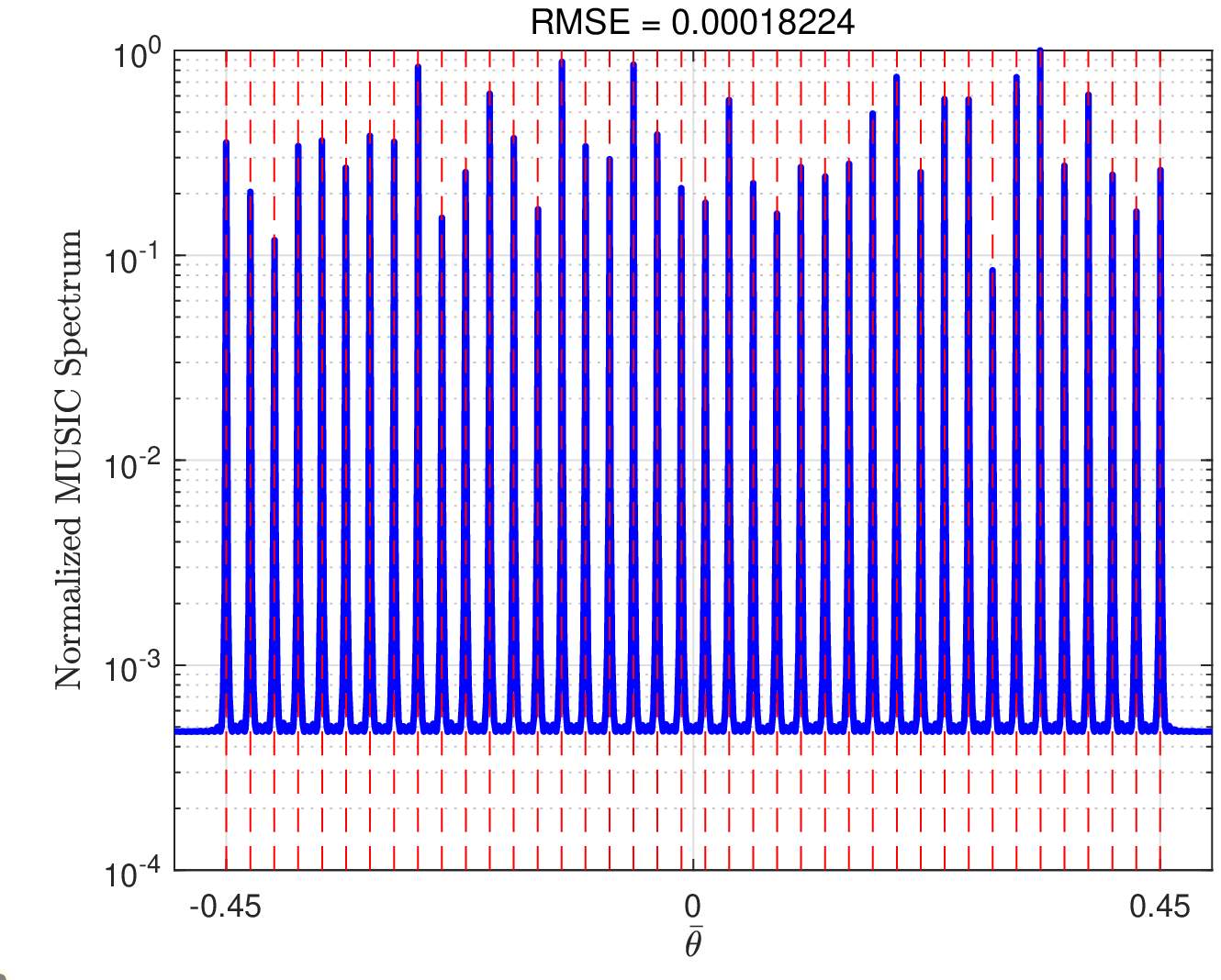}
										
										{(f)}
										\label{fig:side:a}
									\end{minipage}
									\caption{The MUSIC spectra for six kinds of 23-element arrays when $K=40$ sources are located at $\bar{\theta}_{k}=-0.45+0.9(k-1)/39,1\leq k\leq 40.$ (a) $(4r+3)$-Type-93. (b) $(4r)$-Type-93. (c) $(4r)$-Type-10. (d) SNA. (e) MISC. (f) New $(4r)$-Type 1.}
									\label{fig:MUSICspectracoupling}
								\end{figure*}
								
								\subsubsection{RMSE Performance}
								
								{The simulation in this part mainly focuses on RMSE performance under different conditions. The fixed parameter setting is ${\rm SNR}=0$ dB, $T=1000$ snapshots, and $K=35$ sources are incident on the array from $\bar{\theta}_k=-0.45+0.9(k-1)/(K-1),1\leq k\leq K$.}
								
								
								{Fig. $\ref{fig:RMSESNRcoupling}$ shows the relationship between the RMSE of the normalized DOA estimation and the SNR. Observe from this figure, we can see that
									our new array yields the best DOA estimation performance over the entire SNR range, and the performance of other arrays is obviously worse than our proposed array. This shows that our new array is more robust to coupling effect. Furthermore, we should note that although the $(4r+3)$-Type array has the largest uDOFs, its DOA estimation is relatively poor when mutual coupling is considered, because of the dense ULA in its configuration.}
								
								
								\begin{figure}[htbp]
									\centering
									\includegraphics[width=8.5cm]{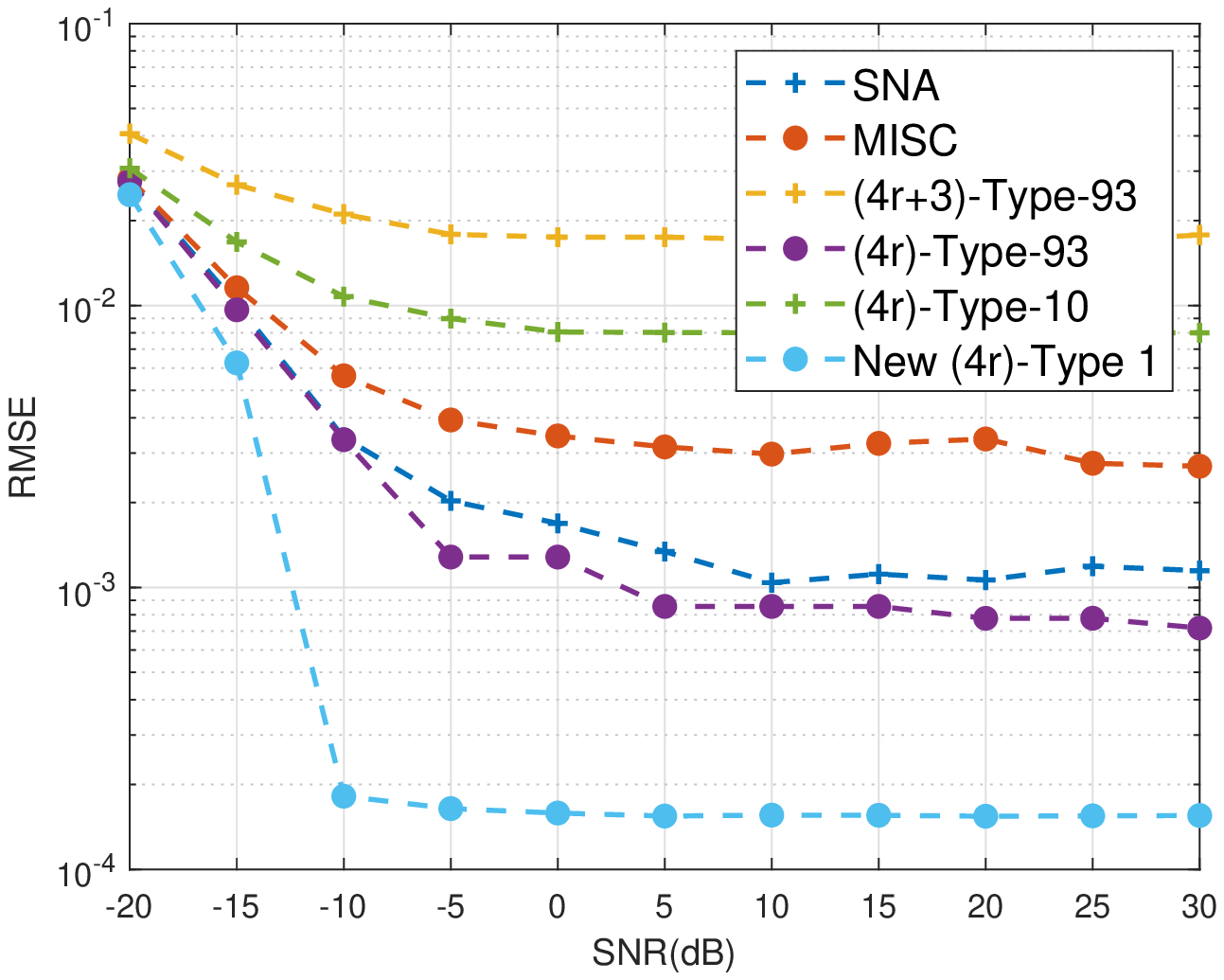}\\
									\caption{RMSE of normalized DOA estimates versus the SNR in the presence of mutual coupling.}
									\label{fig:RMSESNRcoupling}
								\end{figure}
								
								{Fig. $\ref{fig:RMSEsnapshotscoupling}$ illustrates the relationship between the RMSE of the normalized DOA estimation and the number of snapshots.} It is observed that, as the number of snapshots increases, the RMSE of all arrays are reduced rapidly until $T$ reaches about $900$, except for $(4r+3)$-Type and $(4r)$-Type-10 arrays. This is because the both arrays have higher coupling leakage compared to other arrays. Especially, our new array has the  lowest RMSE  than other arrays over the entire snapshots range.
								
								
								\begin{figure}[htbp]
									\centering
									\includegraphics[width=8.5cm]{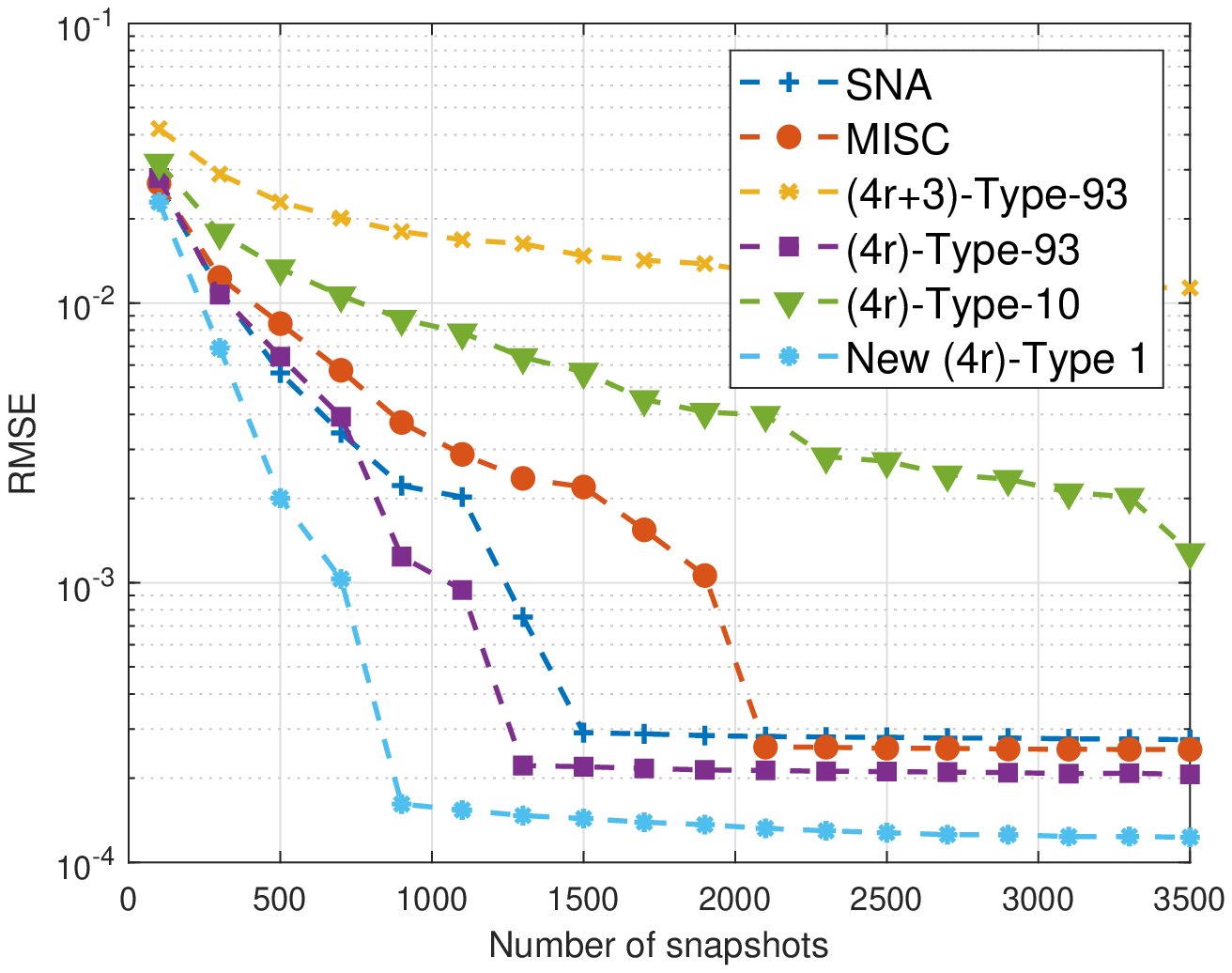}\\
									\caption{RMSE of normalized DOA estimates versus the number of  snapshots in the presence of mutual coupling.}
									\label{fig:RMSEsnapshotscoupling}
								\end{figure}

								{Fig. $\ref{fig:RMSEouhexishucoupling}$ depicts the RMSE curves versus $|c_{1}|$.}
								For any array geometry, the corresponding RMSE increases along with the increase of $|c_{1}|$. That is because a higher value of $|c_{1}|$ introduces more severe mutual coupling effect. When $|c_1|$ is less than 0.7, our new array yields the best performance while the $(4r+3)$-Type array achieves the worst performance, which is because the estimation accuracy is severely affected by uDOFs and mutual coupling effects together. 
								
								
								\begin{figure}[htbp]
									\centering
									\includegraphics[width=8.5cm]{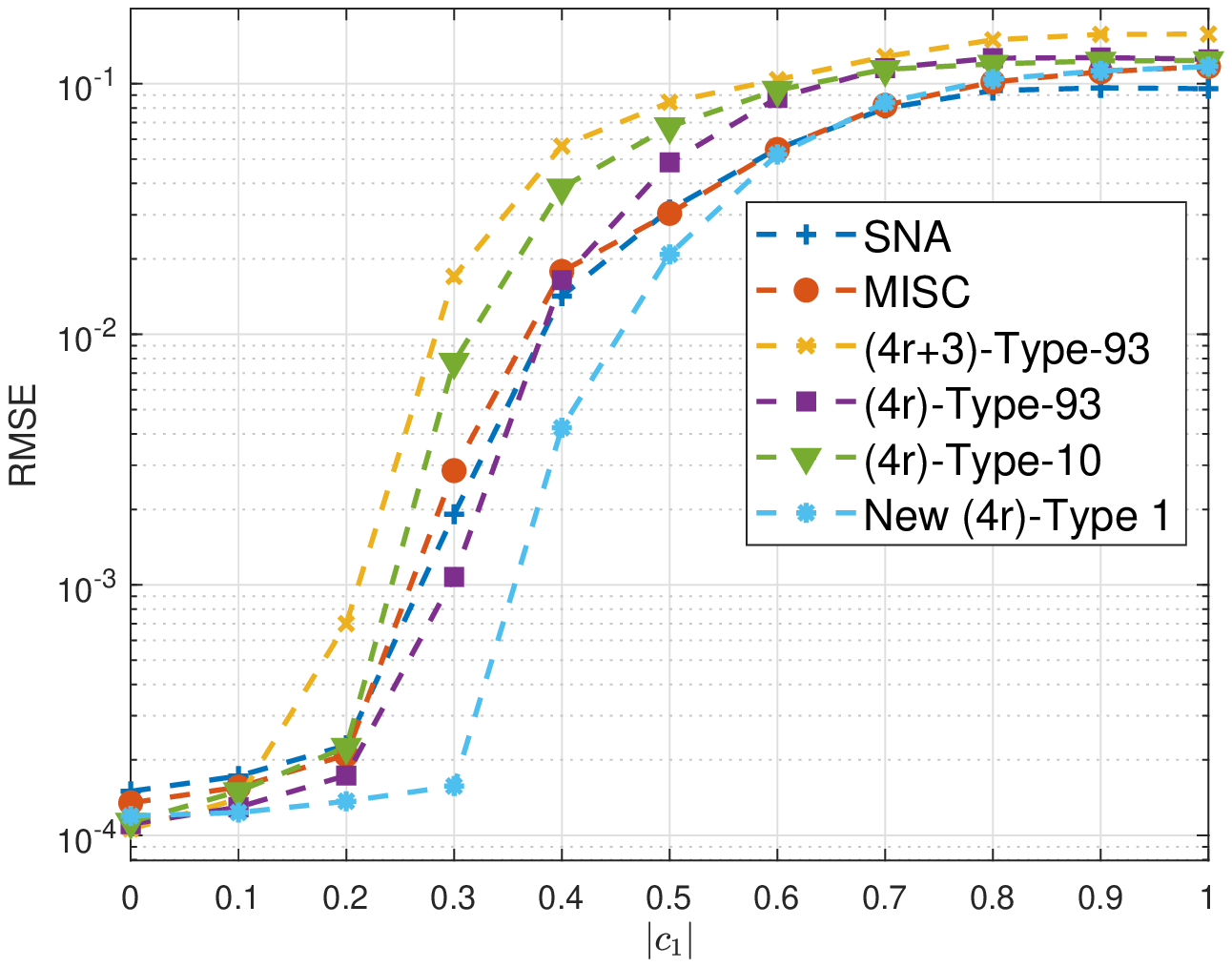}\\
									\caption{RMSE of normalized DOA estimates versus $|c_{1}|$.}
									\label{fig:RMSEouhexishucoupling}
								\end{figure}

								All the above estimations are performed for arrays with $K=23$ sensors. But we still want to know the RMSE of the normalized DOA estimates versus different  number $K$ of sources, which is depicted in
								Fig. $\ref{fig:RMSEzhenyuanshucoupling}$. When $K$ is small, the $(4r+3)$-Type-93 array has the minimum RMSE, but its RMSE curve increases rapidly as $K$ exceeds 20, which implies that the angle measurement accuracy decreases rapidly.
								{When the number of signal sources is less than $35$, the DOA estimation performance of our proposed array is at a stable level. Although it deteriorates when the number of signal sources is greater than 35, it still has better RMSE performance than some arrays.}
								
								
								\begin{figure}
									\centering
									\includegraphics[width=8.5cm]{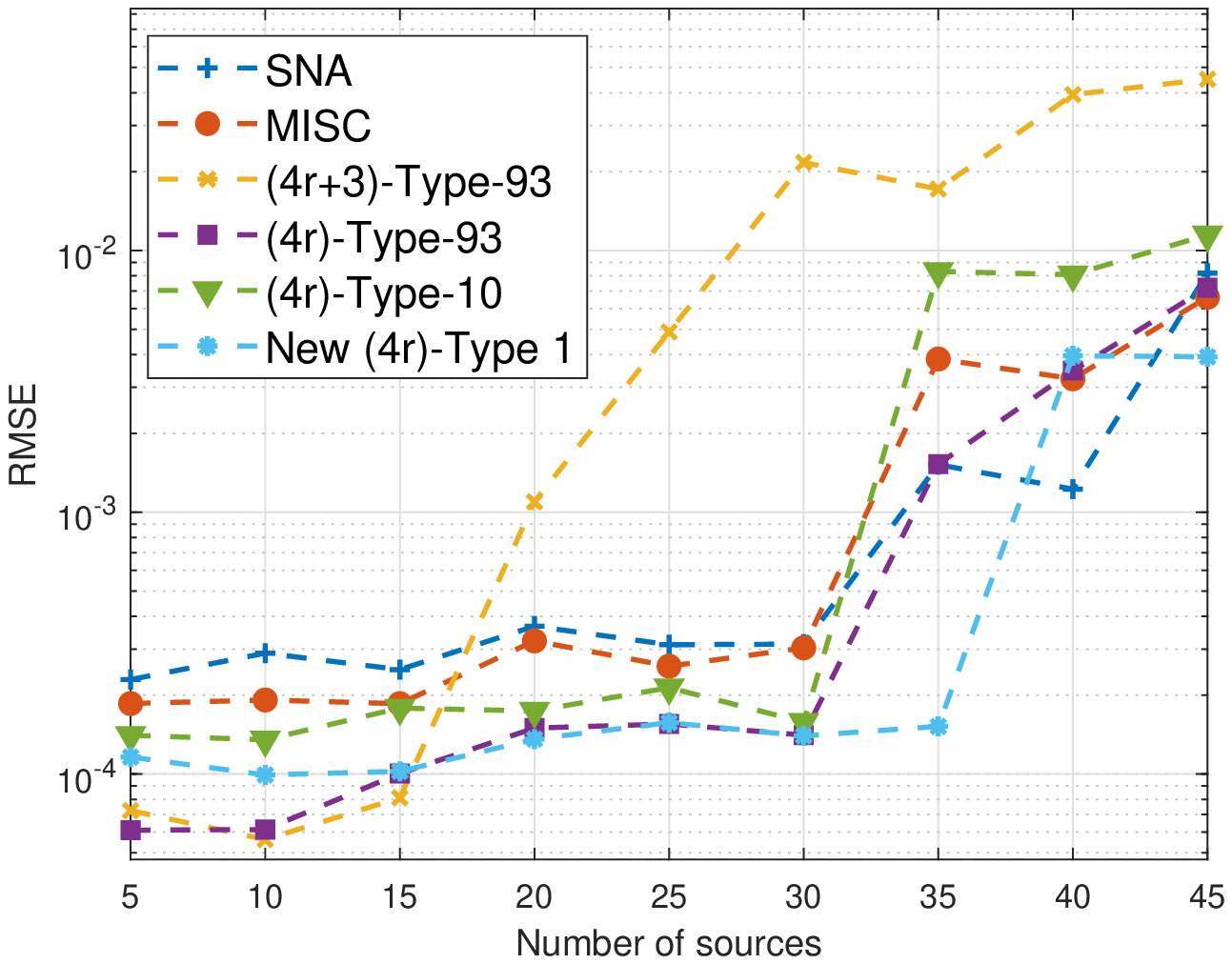}\\
									\caption{RMSE of normalized DOA estimates versus the number of sources in the presence of mutual coupling.}
									\label{fig:RMSEzhenyuanshucoupling}
								\end{figure}

								\section{Conclusions}\label{conclusion}
								In this paper, we have proposed some new array configurations which combine the advantages of both low redundancy ratio and low mutual coupling of known sensor arrays. Since all the known good LRAs were obtained under a common structure pattern with a fixed restriction, which have been extensively studied. In order to reach our goal, we gave a new restriction $s_1+s_2=c$ on the common array pattern, and fortunately obtained some new good LRAs. Especially we got a class of arrays achieving $R<1.5$ and $\omega(1)=1$, which can not be done for the existing LRA configurations. {Numerical results also verified the superiority of this new array to the existing LRAs.}
								
								\section*{Acknowledgments}
								The first and second authors want to thank  Prof. Cuiling Fan and Prof. Yang Yang for their interesting discussions and several exciting suggestions along this work.

								%
								
								\begin{appendices}
									\section{proof of lemma \ref{hole-free}}\label{appendix A}
									\begin{proof}
										{According to the definition of ${\cal D}_{New}$, we know that the elements in ${\cal D}_{New}$ appear symmetrically about 0. In order to prove that the array is hole-free, we only need to prove that ${\cal D}_{New}$ contains ${\mathbb F}=\{1,2,3,\ldots,L\}$.}
										
										{The following $N-1$ positive difference sets can be obtained by subtracting the elements in ${\mathbb S}_{New}=\{s_1,\ldots,s_N\}$}:
										\begin{align*}
											&T_{1}=  \{s_{2}-s_{1},s_{3}-s_{1},\ldots,s_{N}-s_{1}\}, \\
											&T_{2}=  \{s_{3}-s_{2},s_{4}-s_{2},\ldots,s_{N}-s_{2}\},\\
											&.....................\\
											&T_{N-2}=  \{s_{N-1}-s_{N-2},s_{N}-s_{N-2}\},\\
											&T_{N-1}=  \{s_{N}-s_{N-1}\}.
										\end{align*}
										Let $T_{0}=\bigcup\limits_{j=1}^{N-1}{T}_j$ and ${\mathbb F}=\bigcup\limits_{i=0}^{3r+k-1}{\mathbb F}_i$, where ${\mathbb F}_i$ is the consecutive lags of ${\mathbb F}$ defined as:
										\begin{equation*}
											{\mathbb F}_i=\{4ri+1,4ri+2,\ldots,4r(i+1)\},~{\rm for}~0\leq i\leq 3r+k-2,
										\end{equation*}
										\begin{equation*}
											{\mathbb F}_{3r+k-1}=\{4r(3r+k-1)+1,\ldots,4r(3r+k)-2\}.
										\end{equation*}
										{Next, we will show that any ${\mathbb F}_i$ is contained in the union of some $T_{j}$.}
										We will do this case by case:
										\begin{itemize}
											\item when $i=0,$ \begin{align*}{\mathbb F}_0\subset& T_{1}\cup\cdots\cup T_{r}\cup T_{2r+1}\cup T_{4r+k}\cup T_{5r+k-2}\\
												&\cup T_{5r+k-1}\cup T_{5r+k+1};
											\end{align*}
											\item when $1\leq i\leq \lceil\frac{r}{2}\rceil-2,$
											\begin{align*}{\mathbb F}_i\subset & T_{1}\cup\cdots\cup T_{r}\cup T_{2r+1-2i}\cup\cdots\cup T_{2r+1}\\
												&\cup T_{4r+k-i}\cup\cdots\cup T_{4r+k}\\
												&\cup T_{5r+k-2(i+1)}\cup\cdots\cup T_{5r+k-2i};\end{align*}
											\item when $i=\lceil\frac{r}{2}\rceil-1,$
											\begin{align*}{\mathbb F}_i\subset& T_{1}\cup\cdots\cup T_{r}\cup T_{r+2+\lceil\frac{r+1}{2}\rceil-\lceil\frac{r}{2}\rceil}\cup\cdots\cup T_{2r+1}\\
												&\cup T_{4r+k-i}\cup\cdots\cup T_{4r+k+\lceil\frac{r+1}{2}\rceil-\lceil\frac{r}{2}\rceil};\end{align*}
											\item when $\lceil\frac{r}{2}\rceil\leq i\leq 2r+k-2,$
											\[{\mathbb F}_i\subset T_1\cup\cdots\cup T_{2r+1}\cup T_{4r+k-i}\cup\cdots\cup T_{4r+k+\lceil\frac{r}{2}\rceil-i};\]
											\item when $2r+k-1\leq i\leq 2r+k-2+\lceil\frac{r+1}{2}\rceil,$
											\[{\mathbb F}_i\subset T_1\cup\cdots\cup T_{4r+k-i+\lceil\frac{r+1}{2}\rceil};\]
											\item when $i=2r+k-2+\lceil\frac{r+1}{2}\rceil+k'$ for $1\leq k'\leq \lceil\frac{r}{2}\rceil,$
											\[{\mathbb F}_i\subset T_1\cup\cdots\cup T_{2(r-k')+2+\lceil\frac{r}{2}\rceil-\lceil\frac{r+1}{2}\rceil}.\]
										\end{itemize}
										Therefore, the difference co-array of our new array is a  hole-free ULA, i.e., ${\cal D}_{New}=[-L,L]$ with $L=12r^2+4rk-2$.
										
										The proof of another method of item by item can be found at https://arxiv.org/abs/2208.05263.
										Furthermore, we have omitted the proofs for the (4r)-Type 2 arrays and the (4r)-Type 3 arrays since their proofs are similar to that of the (4r)-Type 1 array.
									\end{proof}
									Now we use an example to illustrate the procedure of proof in Lemma \ref{hole-free}.
									\begin{example}
										Let $N=18$, i.e., $r=3,k=0$. The structure of our new array with 18-sensor is
										\begin{align*}
											\{&0,1,3,6,11,16,21,33,45,57,69,81,88,95,97,102,\\
											& 104,106\}.
										\end{align*}
										Thus we obtain $17$ positive difference sets  as follows:
										\begin{align*}
											&T_{1}= \{1,3,6,11,16,21,33,45,57,69,81,88,95,97,\\
											&~~~~~~~102,104,106\},\\
											&T_{2}= \{2,5,10,15,20,32,44,56,68,80,87,94,96,101,\\
											&~~~~~~~103,105\},\\
											&T_{3}= \{3,8,13,18,30,42,54,66,78,85,92,94,99,101,103\},\\
											&T_{4}= \{5,10,15,27,39,51,63,75,82,89,91,96,98,100\},\\
											&T_{5}= \{5,10,22,34,46,58,70,77,84,86,91,93,95\},\\
											&T_{6}= \{5,17,29,41,53,65,72,79,81,86,88,90\},\\
											&T_{7}= \{12,24,36,48,60,67,74,76,81,83,85\},\\
											&T_{8}= \{12,24,36,48,55,62,64,69,71,73\},\\
											&T_{9}= \{12,24,36,43,50,52,57,59,61\},\\
											&T_{10}= \{12,24,31,38,40,45,47,49\},\\
											&T_{11}= \{12,19,26,28,33,35,37\},\\
											&T_{12}= \{7,14,16,21,23,25\},\\
											&T_{13}= \{7,9,14,16,18\},\\
											&T_{14}= \{2,7,9,11\},\\
											&T_{15}= \{5,7,9\},\\
											&T_{16}= \{2,4\},\\
											&T_{17}= \{2\}.
										\end{align*}
										We want to prove that ${\cal D}_{New}^+=[0,106]$. Define ${\mathbb F}_i=\{12i+1,12i+2,\ldots,12(i+1)\}~(0\leq i\leq7)$ and ${\mathbb F}_{8}=\{97,98,\ldots,106\}$. It is easy to obtain
										\begin{align*}
											&\{1,2,\ldots,12\}\subset T_{1}\cup T_{2}\cup T_{3}\cup T_{7}\cup T_{12}\cup T_{13}\cup T_{16},\\
											&\{13,14,\ldots,24\}\subset T_{1}\cup T_{2}\cup T_{3}\cup T_{5}\cup T_{6}\cup T_{7}\\
											&~~~~~~~~~~~~~~~~~~~~~~\cup T_{11}\cup T_{12},\\
											&\{25,26,\ldots,36\}\subset T_{1}\cup T_{2}\cup\cdots\cup T_{7}\cup T_{10}\cup T_{11}\cup T_{12},\\
											&\{37,38,\ldots,48\}\subset T_{1}\cup T_{2}\cup\cdots\cup T_{7}\cup T_{9}\cup T_{10}\cup T_{11},\\
											&\{49,50,\ldots,60\}\subset T_{1}\cup T_{2}\cup\cdots\cup T_{7}\cup T_{8}\cup T_{8}\cup T_{10},\\
											&\{61,62,\ldots,72\}\subset T_{1}\cup T_{2}\cup\cdots\cup T_{7}\cup T_{8}\cup T_{9},\\
											&\{73,74,\ldots,84\}\subset T_{1}\cup T_{2}\cup\cdots\cup T_{7}\cup T_{8},\\
											&\{85,86,\ldots,96\}\subset T_{1}\cup T_{2}\cup\cdots\cup T_{6},\\
											&\{97,98,\ldots,106\}\subset T_{1}\cup T_{2}\cup\cdots\cup T_{4};
										\end{align*}
										which coincide with the process of proof in Lemma \ref{hole-free}. Thus we obtain that the new array is hole-free.
									\end{example}

									
								\end{appendices}

							\end{document}